\documentclass[onefignum,onetabnum]{siamart171218}

\usepackage{graphicx, enumerate, url, pifont}

\usepackage{tabularx}
\usepackage{mathrsfs} 
\usepackage{mathtools, amssymb}
\usepackage{pifont}

\usepackage{kbordermatrix}

\usepackage{tikz}
\usetikzlibrary{matrix,arrows,automata}

\usepackage{tkz-graph}

\usepackage{caption}
\usepackage{subcaption} 

\newsiamremark{example}{Example}

\DeclareMathOperator{\grad}{grad} 
\DeclareMathOperator{\curl}{curl} 
\DeclareMathOperator{\dive}{div} 
\DeclareMathOperator{\inflow}{inflow} 
\DeclareMathOperator{\outflow}{outflow} 
\DeclareMathOperator{\sgn}{sgn} 
\DeclareMathOperator{\diag}{diag} 
\DeclareMathOperator{\im}{im} 
 
\DeclareMathOperator{\rank}{rank} 
\DeclareMathOperator{\nullity}{nullity}

\newcommand{\tp}{{\scriptscriptstyle\mathsf{T}}}

\usepackage{adjustbox}
\newcommand{\Wedge}{\mathord{\adjustbox{valign=B,totalheight=.65\baselineskip}{$\bigwedge$}}}

\headers{Hodge Laplacians on Graphs}{L.-H.~Lim}

\title{Hodge Laplacians on Graphs\thanks{Submitted to the editors \today.
\funding{This work is supported by DARPA D15AP00109, NSF IIS 1546413, a DARPA Director's Fellowship, and the Eckhardt Faculty Fund.}}}

\author{Lek-Heng Lim\thanks{Computational and Applied Mathematics Initiative, Department of Statistics, The University of Chicago, Chicago, IL (\email{lekheng@galton.uchicago.edu})}
}

\ifpdf
\hypersetup{
  pdftitle={Hodge Laplacians on Graphs},
  pdfauthor={Lek-Heng Lim}
}
\fi

\begin{document}

\maketitle

\begin{abstract}
This is an elementary introduction to the Hodge Laplacian on a graph, a higher-order generalization of the graph Laplacian. We will discuss basic properties including cohomology and Hodge theory. The main feature of our approach is simplicity, requiring only knowledge of linear algebra and graph theory. We have also isolated the algebra from the topology to show that a large part of cohomology and Hodge theory is nothing more than the linear algebra of matrices satisfying $AB = 0$. For the remaining topological aspect, we cast our discussions entirely in terms of graphs as opposed to less-familiar topological objects like simplicial complexes.
\end{abstract}



\section{Introduction}

The primary goal of this article is to introduce readers to the Hodge Laplacian on a graph and discuss some of its properties, notably the Hodge decomposition. To understand its significance, it is inevitable that we will also have to discuss the basic ideas behind cohomology, but we will do so in a way that is as elementary as possible and with a view towards applications in the information sciences. 

If the classical Hodge theory on Riemannian manifolds \cite{H, W} is ``differentiable Hodge theory,''  the Hodge theory on metric spaces  \cite{BSSS, SS} ``continuous Hodge theory,'' and the Hodge theory on simplicial complexes \cite{D1, Eck} ``discrete Hodge theory,'' then the version here may be considered ``graph-theoretic Hodge theory.''

Unlike physical problems arising from areas such as continuum mechanics or electromagnetics, where the differentiable Hodge--de Rham theory has been applied with great efficacy for both modeling and computations \cite{feec, AFW, DGLZ, KE, MMOC, TJ, TLHD}, those arising from data analytic  applications are likely to be far less structured \cite{Austin, link, deep, CMOP, COP, data, DMV, synchron, HKW, JLYY, resist, brain,  ODO, dim, group, denoise, XHJYYL, cryo, coverage}. Often one could at best  assume some weak notion of proximity of data points. The Hodge theory introduced in this article requires nothing more than the data set having the structure of an undirected graph and is conceivably more suitable for non-physical applications such as those arising from the biological or information sciences (see Section~\ref{sec:app}).

Our simple take on  cohomology and Hodge theory requires only linear algebra and graph theory. In our approach, we have isolated the algebra from the topology to show that a large part of cohomology and Hodge theory is nothing more than the linear algebra of matrices satisfying $AB = 0$. For the remaining topological aspect, we cast our discussions entirely in terms of graphs as opposed to less-familiar topological objects like simplicial complexes. We believe that by putting these in a simple framework, we could facilitate the development of  applications as well as communication with practitioners who may not otherwise see the utility of these notions.

We write with a view towards readers whose main interests may lie in machine learning, matrix computations, numerical PDEs, optimization, statistics, or theory of computing, but have a casual interest in the topic and may perhaps want to explore potential applications in their respective fields. To enhance the pedagogical value of this article, we have provided complete proofs and fully worked-out examples in Section~\ref{sec:proofs}.

The occasional whimsical section headings are inspired by \cite{bumper, ped1, ped2, DGC, onefoot}.

\section{Cohomology and Hodge theory for pedestrians}\label{sec:ped}

We will present in this section what we hope is the world's most elementary approach towards \textit{cohomology} and \textit{Hodge theory}, requiring only linear algebra. 

\subsection{Cohomology on a bumper sticker}\label{sec:coho}

Given two matrices $A \in \mathbb{R}^{m\times n}$ and $B \in \mathbb{R}^{n \times p}$ satisfying the property that
\begin{equation}\label{eq:null}
AB =0,
\end{equation}
a property equivalent to
\[
\im(B) \subseteq \ker(A),
\]
the \textit{cohomology group} with respect to $A$ and $B$ is the quotient vector space
\[
\ker(A)/\im(B),
\]
and its elements are called \textit{cohomology classes}. The word `group' here refers to the structure of $\ker(A)/\im(B)$ as an abelian group under addition.

We have fudged a bit because we haven't yet defined the matrices $A$ and $B$. Cohomology usually refers to a special case where $A$ and $B$ are certain matrices with topological meaning, as we will define in Section~\ref{sec:graph}.

\subsection{Harmonic representative}\label{sec:harm}

The definition in the previous section is plenty simple, provided the reader knows what a quotient vector space is, but can it be further simplified? For instance, can we do away with quotient spaces and equivalence classes\footnote{Practitioners tend to dislike working with equivalence classes of objects. One reason is that these are often tricky to implement in a computer program.} and define cohomology classes as actual vectors in $\mathbb{R}^n$?

Note that an element in $\ker(A)/\im(B)$ is a set of vectors
\[
x + \im(B) \coloneqq \{ x+y \in \mathbb{R}^n :  y \in \im(B)\}
\]
for some $x \in \ker(A)$. We may avoid such equivalence classes if we could choose an $x_H \in x + \im(B)$ in some unique way to represent the entire set. A standard way to do this is to pick $x_H$ so that it is orthogonal to every other vector in $\im(B)$. Since $\im(B)^\perp = \ker(B^*)$, this is equivalent to requiring that $x_H \in \ker(B^*)$. Hence we should pick an $x_H \in \ker(A) \cap \ker(B^*)$. Such an $x_H$ is called a \textit{harmonic representative} of the cohomology class $x + \im(B)$.

The map that takes the cohomology class $x + \im(B)$ to its unique harmonic representative $x_H$ gives a natural isomorphism of vector spaces (see Theorem~\ref{thm:cohomology})
\begin{equation}\label{eq:natural}
\ker(A)/\im(B) \cong \ker(A) \cap \ker(B^*).
\end{equation}
So we may redefine the cohomology group with respect to $A$ and $B$ to be the subspace $\ker(A) \cap \ker(B^*)$ of $\mathbb{R}^n$, and a cohomology class may now be regarded as an actual vector $x_H \in  \ker(A) \cap \ker(B^*)$.

A word about our notation: $B^*$ denotes the adjoint of the matrix $B$. Usually we  will work over $\mathbb{R}$ with the standard $l^2$-inner product on our spaces and so $B^* = B^\tp $. However we would like to allow for the possibility of working over $\mathbb{C}$ or with other inner products.

\paragraph{Linear algebra interlude} For those familiar with numerical linear algebra, the way we choose a unique harmonic representative $x_H$ to represent a cohomology class $x + \im(B)$ is similar to how we would impose uniqueness on a solution to a linear system of equations by requiring that it has minimum norm among all solutions \cite[Section~5.5]{GMW}. More specifically, the solutions to $Ax=b$ are given by $x_0 + \ker(A)$ where $x_0$ is any particular solution; we impose uniqueness by requiring that $x_0 \in \ker(A)^\perp = \im(A^*) $, which gives the minimum norm (or pseudoinverse) solution $x_0 = A^\dag b$. The only difference above is that we deal with two matrices $A$ and $B$ instead of a single matrix $A$.

\subsection{Hodge theory on one foot}\label{sec:hodge}

We now explain why an element in $\ker(A) \cap \ker(B^*)$ is called  `harmonic.' Again assume that $AB = 0$, the \textit{Hodge Laplacian} is the matrix
\begin{equation}\label{eq:hodge}
A^*A + BB^*  \in \mathbb{R}^{n \times n}.
\end{equation}
We may show (see Theorem~\ref{thm:hodge}) that
\begin{equation}\label{eq:inter}
\ker(A^*A + BB^*) = \ker(A) \cap \ker(B^*).
\end{equation}
So the harmonic representative $x_H$ that we constructed in Section~\ref{sec:harm} is a solution to the \textit{Laplace equation}
\begin{equation}\label{eq:Lap}
(A^*A + BB^*)x = 0.
\end{equation}
Since solutions to the Laplace equation are called \textit{harmonic} functions, this explains the name `harmonic' representative.

With this observation, we see that we could also have defined the cohomology group (with respect to $A$ and $B$) as the kernel of the Hodge Laplacian since
\[
\ker(A)/\im(B) \cong \ker(A^*A + BB^*).
\]
We may also show (see Theorem~\ref{thm:hodge}) that there is a \textit{Hodge decomposition}, an orthogonal direct sum decomposition
\begin{equation}\label{eq:decomp}
\mathbb{R}^n = \im(A^*) \oplus \ker(A^*A + BB^*) \oplus \im(B).
\end{equation}
In other words, whenever $AB = 0$, every $x \in \mathbb{R}^n$ can be decomposed uniquely as
\[
x = A^*w + x_H + Bv, \qquad \langle A^*w, x_H \rangle = \langle x_H, Bv \rangle =\langle A^*w, Bv \rangle =0,
\]
for some $v \in \mathbb{R}^p$ and $w \in \mathbb{R}^m$.

Recall the well-known decompositions (sometimes called the Fredholm alternative, see Theorem~\ref{thm:fredholm}) associated with the four fundamental subspaces \cite{Strang} of a matrix  $A \in \mathbb{R}^{m \times n}$,
\begin{equation}\label{eq:fred}
\mathbb{R}^n = \ker(A) \oplus \im(A^*), \quad 
\mathbb{R}^m = \ker(A^*) \oplus \im(A).
\end{equation}
The Hodge decomposition \eqref{eq:decomp} may be viewed as an analogue of \eqref{eq:fred} for a pair of matrices satisfying $AB=0$. In fact, combining \eqref{eq:decomp} and \eqref{eq:fred}, we obtain
\[
\mathbb{R}^n = \rlap{$\overbrace{\phantom{\im(A^*) \oplus \ker(A^*A + BB^*)}}^{\ker(B^*)}$}\im(A^*) \oplus \underbrace{\ker(A^*A + BB^*) \oplus \im(B)}_{\ker(A)}.
\]
The intersection of $\ker(A)$ and $\ker(B^*)$ gives $\ker(A^*A + BB^*)$, confirming \eqref{eq:inter}. Since $A^*A + BB^*$ is Hermitian, it also follows that
\begin{equation}\label{eq:sum}
\im(A^*A + BB^*)  = \im(A^*) \oplus \im(B).
\end{equation}

For the special case when $A$ is an arbitrary matrix and $B=0$, the Hodge decomposition \eqref{eq:decomp} becomes
\begin{equation}\label{eq:decomp0}
\mathbb{R}^n =\im(A^*) \oplus \ker(A^*A),
\end{equation}
which may also be deduced directly from the Fredholm alternative \eqref{eq:fred} since
\begin{equation}\label{eq:normal}
\ker(A^* A) = \ker(A).
\end{equation}

\paragraph{Linear algebra interlude}  To paint an analogy like that in the last paragraph of Section~\ref{sec:harm}, our characterization of cohomology classes as solutions to the Laplace equation \eqref{eq:Lap} is similar to the characterization of solutions to a least squares problem $\min_{x \in \mathbb{R}^n} \lVert Ax - b \rVert$ as solutions to its normal equation $A^*A x = A^* b$ \cite[Section~6.3]{GMW}. Again the only difference is that here we deal with two matrices instead of just one.

\subsection{Terminologies}\label{sec:terms}

One obstacle that the (impatient) beginner often faces when learning cohomology is the considerable number of scary-sounding terminologies that we have by-and-large avoided in the treatment above.

In Table~\ref{tab:terms1}, we summarize some commonly used terminologies for objects in Sections~\ref{sec:coho}, \ref{sec:harm},  and \ref{sec:hodge}. Their precise meanings will be given in Sections~\ref{sec:graph} and \ref{sec:HO}, with an updated version of this table appearing as Table~\ref{tab:terms2}.
\begin{table}[h!]
\centering
\caption{Topological jargons (first pass)}
\label{tab:terms1}
\vspace*{-1.5ex}
\renewcommand{\arraystretch}{1.25}
\begin{tabular}{|l|l|}
\hline
\textsc{name} & \textsc{meaning}\\
\hline
coboundary maps & $A \in \mathbb{R}^{m\times n}$, $B \in \mathbb{R}^{n \times p}$\\
cochains & elements $x \in \mathbb{R}^n$\\
cochain complex & $\mathbb{R}^p \xrightarrow{B} \mathbb{R}^n \xrightarrow{A} \mathbb{R}^m$\\
cocycles & elements of $\ker(A)$\\
coboundaries & elements of $\im(B)$\\
cohomology classes & elements of $ \ker(A)/\im(B)$\\
harmonic cochains & elements of $\ker(A^*A + BB^*) $\\
Betti numbers & $\dim \ker(A^*A + BB^*) $\\
Hodge Laplacians & $A^*A + BB^* \in \mathbb{R}^{n\times n}$\\
$x$ is closed & $Ax = 0$\\
$x$ is exact & $x = Bv$ for some $v \in \mathbb{R}^p$\\
$x$ is coclosed & $B^*x = 0$\\
$x$ is coexact & $x = A^*w$ for some $w \in \mathbb{R}^m$\\
$x$ is harmonic & $(A^*A + BB^*)x =0 $\\
\hline
\end{tabular}
\end{table}
As the reader can see, there is some amount of redundancy in these terminologies; e.g., saying that a cochain is exact is the same as saying that it is a coboundary. This can sometimes add to the confusion for a beginner. It is easiest to just remember equations and disregard jargons. When people say things like `a cochain is harmonic if and only if it is closed and coclosed,' they are just verbalizing \eqref{eq:inter}.

In summary, we saw \textit{three different ways of defining cohomology}: If $A$ and $B$ are matrices satisfying $AB = 0$, then the cohomology group with respect to $A$ and $B$ may be taken to be any one of the following,
\begin{equation}\label{eq:cohomology}
\ker(A) / \im (B), \qquad
\ker(A) \cap \ker (B^*), \qquad
\ker(A^*A + BB^*).
\end{equation}

For readers who may have heard of the term \textit{homology}, that can be defined just by taking adjoints. Note that if $AB = 0$, then $B^* A^* = 0$ and we may let $B^*$ and $A^*$ play the role of $A$ and $B$ respectively. The homology group with respect to $A$ and $B$ may then be taken to be any one of the following,
\begin{equation}\label{eq:homology}
\ker(B^*) / \im (A^*), \qquad
\ker(B^*) \cap \ker (A), \qquad
\ker(BB^* + A^*A).
\end{equation}
As we can see, the last two spaces in \eqref{eq:cohomology} and \eqref{eq:homology} are identical, i.e., there is no difference between cohomology and homology in our context (see Theorem~\ref{thm:cohomology} for a proof and Section~\ref{sec:caveats} for caveats).

\section{Coboundary operators and Hodge Laplacians on graphs}\label{sec:graph}

The way we discussed cohomology and Hodge theory in Section~\ref{sec:ped} relies solely on the linear algebra of operators satisfying $AB=0$; this is the `algebraic side' of the subject. There is also a `topological side' that is just one step away, obtained by imposing the requirement that $A$ and $B$ be \emph{coboundary operators}. Readers may remember from vector calculus identities like $\curl \grad = 0$ or $\dive \curl = 0$ in $\mathbb{R}^3$ --- these are in fact pertinent examples of when $AB = 0$ naturally arises; and as we will soon see, $\dive, \grad, \curl$ are our most basic examples of coboundary operators. Restricting our choices of $A$ and $B$ to coboundary operators allows us to attach topological meanings to the objects in Section~\ref{sec:ped}.

Just like the last section requires nothing more than elementary linear algebra, this section requires nothing more than elementary graph theory. We will discuss simplicial complexes (family of subsets of vertices), cochains (functions on a graph), and coboundary operators (operators on functions on a graph) --- all in the context of the simplest type of graphs: undirected, unweighted, no loops, no multiple edges.

\subsection{Graphs}\label{sec:cliques}

Let $G = (V,E)$ be an undirected graph where $V\coloneqq \{1,\dots,n\}$ is a finite set of vertices and $E \subseteq \binom{V}{2}$ is the set\footnote{Henceforth $\binom{V}{k}$ denotes the set of all $k$-element subsets of $V$. In particular $E$ is not a multiset since our graphs have no loops nor multiple edges.} of edges. Note that once we have specified $G$, we  automatically get \textit{cliques} of higher order --- for example, the set of triangles or \textit{$3$-cliques}  $T \subseteq \binom{V}{3}$ is defined by
\[
\{i,j,k\} \in T \quad \text{iff} \quad \{i,j\}, \{ i, k \}, \{j,k\}  \in E.
\]
More generally the set of \textit{$k$-cliques} $K_k(G) \subseteq \binom{V}{k}$ is defined by
\[
\{i_1,\dots, i_k \} \in K_k(G) \quad \text{iff} \quad \{i_p,i_q\} \in E \; \text{for all} \; 1 \le p < q \le k,
\]
i.e., all pairs of vertices in $\{i_1,\dots, i_k \}$ are in $E$. Clearly, specifying $V$ and $E$ uniquely determines $K_k(G)$ for all $k \ge 3$. In particular we have
\[
K_1(G) = V, \quad K_2(G) = E, \quad K_3(G) = T.
\]

In topological parlance, a nonempty family $K$ of finite subsets of a set $V$ is  called a \textit{simplicial complex} (more accurately, an abstract simplicial complex) if for any set $S$ in $K$, every $S' \subseteq S$ also belongs to $K$. Evidently the set comprising all cliques of a graph $G$,
\[
K(G) \coloneqq \bigcup\nolimits_{k=1}^{\omega(G)} K_k(G),
\]
is a simplicial complex and is called the \textit{clique complex} of the graph $G$. The \textit{clique number} $\omega(G)$ is the number of vertices in a  largest clique of $G$.

There are abstract simplicial complexes that are not clique complexes of graphs. For example, we may just exclude cliques of larger sizes --- $\bigcup\nolimits_{k=1}^{m} K_k(G)$ is still an abstract simplicial complex for any $m =3,\dots,\omega(G)-1,$ but it would not in general be a clique complex of a graph.

\subsection{Functions on a graph}\label{sec:cochian}

Given a graph $G = (V,E)$, we will consider real-valued functions on its vertices $f: V \to \mathbb{R}$. We will also consider real-valued functions on $E$ and $T$ and $K_k(G)$ in general but we shall require them to be \textit{alternating}. By an alternating function on $E$, we mean one of the form $X : V \times V \to \mathbb{R}$ where
\[
X(i,j) = - X(j,i)
\]
for all $\{i,j\} \in E$, and
\[
X(i,j) = 0
\]
for all $\{i,j\} \not\in E$. An alternating function on $T$ is one of the form $\Phi :V \times V \times V \to \mathbb{R}$ where
\[
\Phi(i,j,k) =  \Phi (j,k,i) = \Phi (k,i,j)  = -\Phi(j,i,k) = - \Phi (i,k,j) = - \Phi(k,j,i) 
\]
for all $\{i,j,k\} \in T$, and
\[
\Phi(i,j,k) = 0
\]
for all $\{i,j,k\} \not\in T$. More generally, an alternating function is one where permutation of its arguments has the effect of changing its value by the sign of the permutation, as we will see in \eqref{eq:altfn}.

In topological parlance, the functions $f, X, \Phi$ are called $0$-, $1$-,  $2$-\textit{cochains}. These are discrete analogues of \textit{differential forms} on manifolds \cite{W}. Those who prefer to view them as such often refer to cochains as \textit{discrete differential forms} \cite{DHLM, DLM, Hira} and in which case, $f, X, \Phi$ are $0$-, $1$-,  $2$-\textit{forms} on $G$.

Observe that a $1$-cochain $X$ is completely specified by the values it takes on the set $\{(i,j) : i < j\}$ and a $2$-cochain $\Phi$ is completely specified by the values it takes on the set $\{(i,j,k) : i < j < k \}$. We may equip the spaces of cochains with inner products, for example, as weighted sums
\begin{gather}\label{eq:inner}
\langle f, g \rangle_V = \sum\nolimits_{i=1}^n w_i f(i)g(i), \qquad \langle X,Y\rangle_E = \sum\nolimits_{i < j} w_{ij} X(i,j)Y(i,j), \\ \langle \Phi, \Psi \rangle_T = \sum\nolimits_{i < j < k} w_{ijk} \Phi(i,j,k) \Psi(i,j,k), \notag
\end{gather}
where the weights $w_i, w_{ij}, w_{ijk} $ are given by any positive values invariant under arbitrary permutation of indices. When they take the constant value $1$, we call it the \textit{standard $L^2$-inner product}. By summing only over the sets\footnote{Our choice is arbitrary; any set that includes each edge or triangle exactly once would also serve the purpose. Each such choice corresponds to a choice of direction or orientation on the elements of $E$ or $T$.} $\{(i,j) : i < j\}$ and $\{(i,j,k) : i < j < k \}$, we count each edge or triangle exactly once in the inner products.

We will denote the Hilbert spaces of $0$-, $1$-, and $2$-cochains as $L^2(V)$, $L^2_\wedge(E)$, $L^2_\wedge(T)$ respectively. The subscript $\wedge$ is intended to indicate `alternating'. Note that $L^2_\wedge(V) = L^2(V)$ since for a function of one argument, being alternating is a vacuous property.  We set $L^2_\wedge(\varnothing ) \coloneqq  \{0 \}$ by convention. The $L^2$ prefix is merely to indicate the presence of an inner product. $L^2$-integrability is never an issue since the spaces $V$, $E$, $T$ are finite sets; e.g., any  function $f:V\to \mathbb{R}$ will be an element of $L^2(V)$ as $\lVert f \rVert_V^2 = \langle f, f \rangle_V$ is just a finite sum and thus finite.

The elements of $L^2_\wedge(E)$ (i.e., $1$-cochains) are well-known in graph theory, often called \textit{edge flows}. While the graphs in this article are always undirected and unweighted, a directed graph is simply one equipped with a choice of edge flow $X \in L^2_\wedge(E)$  --- an undirected edge $\{ i,j\} \in E$ becomes a directed edge $(i,j)$ if $X(i,j) >0$ or $(j,i)$ if $X(i,j) < 0$; and the magnitude of $X(i,j)$ may be taken as the weight of that directed edge. So $ L^2_\wedge(E)$ encodes all weighted directed graphs that have the same underlying undirected graph structure.

\subsection{Operators on functions on a graph}\label{sec:coboundary}

We  consider the graph-theoretic analogues of $\grad$, $\curl$, $\dive$ in multivariate calculus. The \textit{gradient} is the linear operator $\grad : L^2(V) \to L^2_\wedge(E)$ defined by
\[
(\grad f)(i,j) = f(j) -f(i)
\]
for all $\{i,j\} \in E$ and zero otherwise. The \textit{curl} is the linear operator $\curl : L^2_\wedge(E) \to L^2_\wedge(T)$ defined by
\[
(\curl X)(i,j,k) = X(i,j) + X(j,k) + X(k,i)
\]
for all $\{i,j,k\} \in T$ and zero otherwise. The \textit{divergence} is the linear operator $\dive : L^2_\wedge(E) \to L^2(V)$ defined by
\[
(\dive X) (i) = \sum_{j=1}^n \frac{w_{ij}}{w_i}  X(i,j)
\]
for all $i \in V$.

Using these, we may construct other linear operators, notably the well-known \textit{graph Laplacian}, the operator $\Delta_0 : L^2(V) \to L^2(V)$ defined by
\[
\Delta_0 = -\dive  \grad,
\]
which is a graph-theoretic analogue of the Laplace operator (see Lemma~\ref{lem:gl}). Less well-known is the \textit{graph Helmholtzian} \cite{JLYY}, the operator $\Delta_1 : L^2_\wedge(E) \to L^2_\wedge(E)$ defined by
\[
\Delta_1 = - \grad  \dive + \curl^*  \curl,
\]
which is a graph-theoretic analogue of the vector Laplacian. We may derive  (see Lemma~\ref{lem:curl*}) an expression for the adjoint of the curl operator,  $\curl^* : L^2_\wedge(T) \to L^2_\wedge(E)$ is given by
\[
(\curl^* \Phi)(i,j) =  \sum_{k=1}^n \frac{w_{ijk}}{w_{ij}}  \Phi(i,j,k)
\]
for all $\{i,j\} \in E$ and zero otherwise.

The gradient and curl are special cases of \textit{coboundary operators}, discrete analogues of \textit{exterior derivatives}, while the graph Laplacian and Helmholtzian are special cases of \textit{Hodge Laplacians}. 

The matrices $A$ and $B$ that we left unspecified in Section~\ref{sec:ped} are coboundary operators. We may show (see~Theorem~\ref{thm:fund}) that the composition
\begin{equation}\label{eq:cg}
\curl \grad  = 0
\end{equation}
and so setting $A = \curl$ and $B = \grad$ gives us \eqref{eq:null}.

Note that  divergence and gradient are negative adjoints of each other:
\begin{equation}\label{eq:adj}
\dive  = -\grad^*,
\end{equation}
 (see Lemma~\ref{lem:grad*}). With this we get $\Delta_1 = A^*A + BB^*$ as in \eqref{eq:hodge}.

If the inner products on $L^2(V) $ and $L^2_\wedge(E)$ are taken to be the standard $L^2$-inner products, then \eqref{eq:adj} gives $\Delta_0 = B^*B = B^\tp  B$, a well-known expression of the graph Laplacian in terms of vertex-edge incidence matrix $B$. The operators
\[
\grad^* \grad:  L^2(V) \to L^2(V) \quad \text{and} \quad \curl^*  \curl :  L^2_\wedge(E) \to  L^2_\wedge(E)
\]
are sometimes called the \textit{vertex Laplacian} and \textit{edge Laplacian} respectively. The vertex Laplacian is of course just the usual graph Laplacian but note that the edge Laplacian is not the same as the graph Helmholtzian.

\paragraph{Physics interlude}  Take the standard $L^2$-inner products on $L^2(V) $ and $L^2_\wedge(E)$, the divergence of an edge flow at a vertex $i \in V$ may be interpreted as the \textit{netflow},
\begin{equation}\label{eq:inoutflow}
(\dive X) (i) = (\inflow X)(i) - (\outflow X) (i),
\end{equation}
where \textit{inflow} and \textit{outflow} are defined respectively for any $X \in L^2_\wedge(E)$ and any $i \in V$ as
\[
(\inflow X) (i) = \sum\nolimits_{j: X(i,j) < 0} X(i,j),\quad (\outflow X) (i) = \sum\nolimits_{j: X(i,j) > 0} X(i,j).
\]
Sometimes the terms \textit{incoming flux}, \textit{outgoing flux}, \textit{total flux} are used instead of inflow, outflow, netflow.
Figure~{\sc\ref{fig:eg2}} shows two \textit{divergence-free} edge flows, i.e., inflow equals outflow at every vertex.

Let $X \in L^2_\wedge(E)$. A vertex $i \in V$  is called a \textit{sink} of $X$ if $X(i,j) < 0$ for every neighbor $\{i,j \} \in E$ of $i$.  Likewise a vertex $i \in V$ is called a \textit{source} of $X$ if $X(i,j) > 0$ for every neighbor $\{i,j \} \in E$ of $i$. In general, an edge flow may not have any source or sink\footnote{There is an alternative convention that defines $i \in V$ to be a source (resp.\ sink) of $X$ as long as $\dive X(i) > 0$ (resp.\ $\dive X(i) < 0$) but our definition is much more restrictive.}  but if it can be written as
\[
X = - \grad f
\]
for some $f \in L^2(V)$, often called a \textit{potential} function, then $X$ will have the property of flowing from sources (local maxima of $f$) to sinks (local minima of $f$).

\begin{example}\label{eg:curl}
We highlight a common pitfall regarding curl on a graph. Consider $C_3$ and $C_4$, the \textit{cycle graphs} on three and four vertices in Figure~\ref{fig:eg1}.
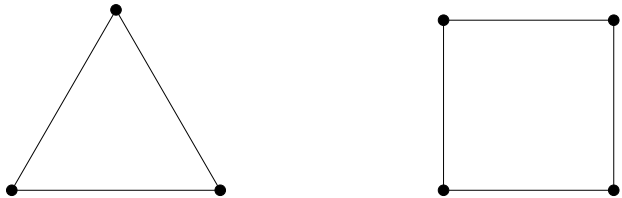
\begin{figure}[h]
\SetGraphUnit{4}
\SetVertexSimple[MinSize = 10pt]
\centering
\scalebox{0.4}{
\begin{tikzpicture}[rotate=90] 
  \Vertices[NoLabel]{circle}{A,B,C}
  \Edges(A,B,C,A)
\end{tikzpicture}}
\hspace{1in}
\scalebox{0.4}{
\begin{tikzpicture}[rotate=45] 
  \Vertices[NoLabel]{circle}{A,B,C,D}
  \Edges(A,B,C,D,A)
\end{tikzpicture}}
\caption{Cycle graphs $C_3$ (\textit{left}) and $C_4$ (\textit{right}).}
\label{fig:eg1}
\end{figure}

Number the vertices and consider the edge flows in Figure~\ref{fig:eg2}. What are the values of their curl? For the one on $C_3$, the answer is $2+2+2 = 6$ as expected. But the answer for the one on $C_4$ is not  $2+2+2+2 = 8$, it is in fact $0$.
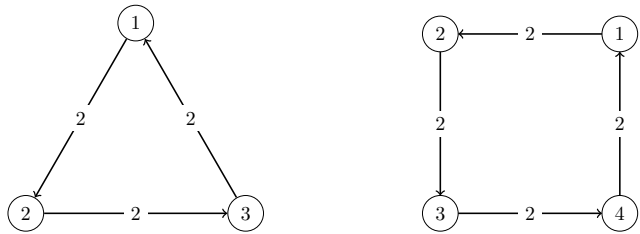
\begin{figure}[h]
\centering
\scalebox{0.75}{
\GraphInit[vstyle=Dijkstra]
\begin{tikzpicture}[scale=2.25,rotate=90]
  \Vertices{circle}{1,2,3}
  \tikzset{EdgeStyle/.style = {->}}
  \Edge[label=$2$](1)(2)
  \Edge[label=$2$](2)(3)
  \Edge[label=$2$](3)(1)
\end{tikzpicture}}
\hspace{0.7in}
\scalebox{0.75}{
\begin{tikzpicture}[scale=2.25,rotate=45]
  \Vertices{circle}{1,2,3,4}
  \tikzset{EdgeStyle/.style = {->}}
  \Edge[label=$2$](1)(2)
  \Edge[label=$2$](2)(3)
  \Edge[label=$2$](3)(4)
  \Edge[label=$2$](4)(1)
\end{tikzpicture}}
\caption{Edge flows on $C_3$ (\textit{left}) and $C_4$ (\textit{right}).}
\label{fig:eg2}
\end{figure}

The second answer may not agree with a physicist's intuitive idea of curl and is a departure from what one would expect in the continuous case. However it is what follows from definition. Let $X \in L^2_\wedge(E(C_3))$ denote the edge flow on $C_3$ in Figure~\ref{fig:eg2}. It is given by
\[
X(1,2) = X(2,3) = X(3,1) =2 = -X(2,1) = -X(3,2) = -X(1,3),
\]
and the curl evaluated at $\{1,2,3\} \in T(C_3)$ is by definition indeed
\[
(\curl X)(1,2,3) = X(1,2) + X(2,3) + X(3,1) = 6.
\]
On the other hand $C_4$ has no $3$-cliques and so $T(C_4) = \varnothing$. By convention $L^2_\wedge(\varnothing ) = \{0 \}$. Hence $\curl : L^2_\wedge(E(C_4)) \to L^2_\wedge(T(C_4)) $ must have $\curl X = 0$ for all $X \in L^2_\wedge(E(C_4)) $ and in particular for the edge flow on the right of Figure~\ref{fig:eg2}.
\end{example}

\begin{table}[h!]
\centering
\caption{Electrodynamics/fluid dynamics jargons}
\label{tab:terms3}
\vspace*{-1.5ex}
\renewcommand{\arraystretch}{1.25}
\begin{tabular}{|l|l|l|}
\hline
\textsc{name} & \textsc{meaning} &\textsc{alternate name(s)}\\
\hline
divergence-free & element of $\ker(\dive)$ & solenoidal \\
curl-free & element of $\ker(\curl)$ & irrotational  \\
vorticity & element of $\im(\curl^*)$ & vector potential\\
conservative & element of $\im (\grad)$ & potential flow \\
harmonic & element of $\ker (\Delta_1)$ &\\
anharmonic & element of $\im (\Delta_1)$ &\\
scalar field & element of $L^2(V)$ & scalar potential\\
vector field & element of $L^2_\wedge(E)$ &\\
\hline
\end{tabular}
\end{table}

\subsection{Helmholtz decomposition for graphs}\label{sec:Helm}

The usual graph Laplacian $\Delta_0 : L^2(V) \to L^2(V)$,
\[
\Delta_0 = -\dive  \grad = \grad^* \grad,
\]
has been an enormously useful construct  in the context of spectral graph theory \cite{Chung, Spielman}, with great impact on many areas. We have nothing more to add except to remark that the Hodge decomposition associated with the graph Laplacian $\Delta_0$ is given by \eqref{eq:decomp0},
\[
L^2(V) = \ker(\Delta_0) \oplus \im(\dive).
\]
Recall from \eqref{eq:normal} that $\ker(\Delta_0) = \ker (\grad)$. Since $\grad f = 0$ iff $f$ is piecewise constant, i.e., constant on each connected component of $G$, the number $\beta_0 (G) \coloneqq \dim \ker (\Delta_0)$ counts the number of connected component of $G$ --- a well known fact in graph theory.

The Hodge decomposition associated with the graph Helmholtzian  $\Delta_1 : L^2_\wedge(E) \to L^2_\wedge(E)$,
\[
\Delta_1 =  - \grad  \dive + \curl^*  \curl = \grad  \grad^* + \curl^*  \curl.
\]
is called the \textit{Helmholtz decomposition}. It says that the space of edge flows admits an orthogonal decomposition into subspaces
\begin{equation}\label{eq:helm1}
L^2_\wedge(E) = \rlap{$\overbrace{\phantom{\im(\curl^*) \oplus \ker(\Delta_1)}}^{\ker(\dive)}$}\im(\curl^*) \oplus \underbrace{\ker(\Delta_1) \oplus \im(\grad)}_{\ker(\curl)},
\end{equation}
and moreover the three subspaces are related via
\begin{equation}\label{eq:helm2}
\ker(\Delta_1)=\ker(\curl)\cap\ker(\dive),\qquad \im(\Delta_1)= \im(\curl^*) \oplus \im(\grad).
\end{equation}
In particular, the first equation is a discrete analogue of the statement ``a vector field is curl-free and divergence-free if and only if it is a harmonic vector field.'' 

There is nothing really special here --- as we saw in Section~\ref{sec:hodge}, any matrices $A$ and $B$ satisfying $AB =0 $ would give such a decomposition: \eqref{eq:helm1} and \eqref{eq:helm2} are indeed just \eqref{eq:decomp}, \eqref{eq:inter}, and \eqref{eq:sum} where $A = \curl$ and $B = \grad$. This is however a case that yields the most interesting applications (see Section~\ref{sec:app} and \cite{CMOP,JLYY}).

\begin{example}[Beautiful Mind problem on graphs]
This is a discrete analogue of a problem\footnote{Due to Dave Bayer \cite{butler}. See Figure~\ref{fig:beau}.} that appeared in a blockbuster movie: Let $G =(V,E)$ be a graph. If $X \in L^2_\wedge(E)$ is curl-free, then is it true that $X$ is a gradient? In other words, if $X \in \ker(\curl)$, must it also be in $\im(\grad)$?
Clearly the converse always holds by \eqref{eq:cg} but from \eqref{eq:helm1}, we know that
\begin{equation}\label{eq:bm}
\ker(\curl) = \ker(\Delta_1) \oplus \im(\grad)
\end{equation}
and so it is not surprising that the answer is generally no. We would like to describe a family of graphs for which the answer is yes.
\begin{figure}[h]
\centering
\includegraphics[clip = true, trim = 35ex 15ex 0 10ex, width=\linewidth]{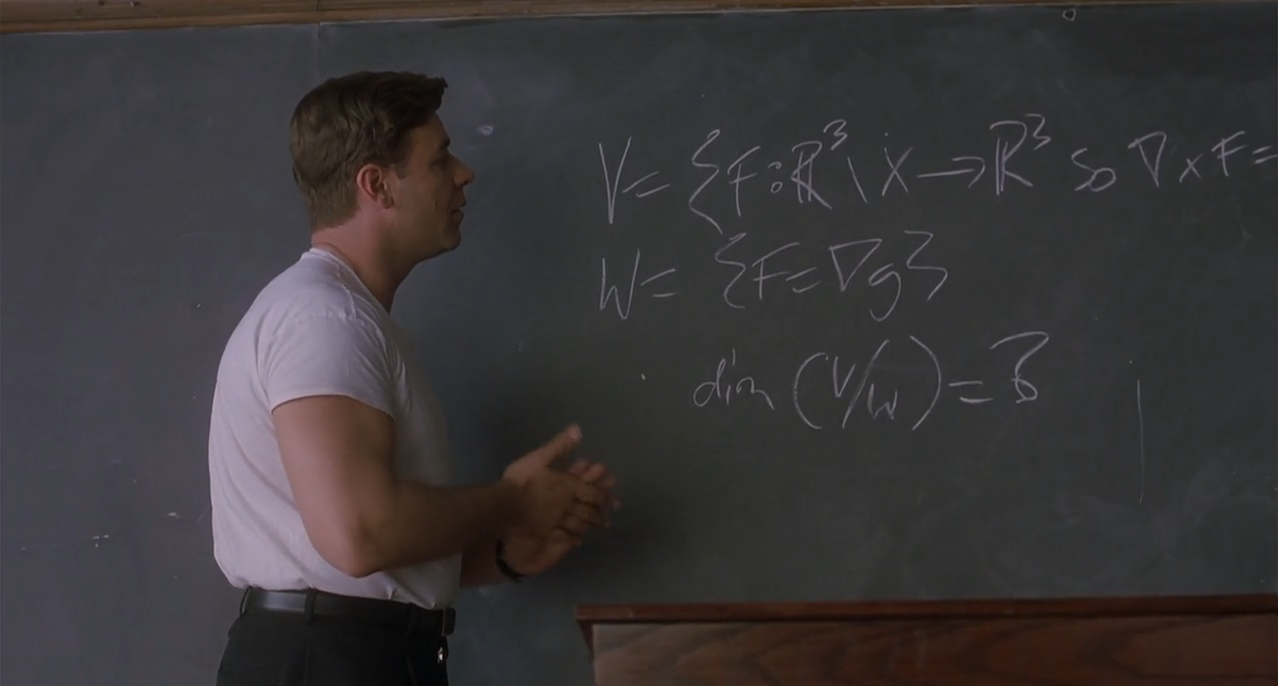}
\caption{Problem from \textit{A Beautiful Mind}: $V = \{ F : \mathbb{R}^3 \setminus X \to \mathbb{R}^3 \text{ so } \nabla \times F = 0 \}$, $W = \{ F = \nabla g \}$, $\dim(V/W) =\ ?$}
\label{fig:beau}
\end{figure}

The edge flow $X \in L^2_\wedge (E (C_4))$ on the right of Figure~\ref{fig:eg2} is an example of one that is curl-free but not a gradient. It is trivially curl-free since $T(C_4) = \varnothing$. It is not a gradient since if $X = \grad f$, then
\[
f(2) - f(1) = 2,\quad f(3) - f(2) = 2,\quad f(4) - f(3) = 2,\quad f(1) - f(4) = 2,
\]
and summing them  gives `$0 = 8$' --- a contradiction. Note that $X$ is also divergence-free by \eqref{eq:inoutflow} since $\inflow  X = \outflow  X$. It is therefore harmonic by \eqref{eq:helm2}, i.e., $X \in \ker(\Delta_1)$ as expected.

Every divergence-free edge flow on $C_4$ must be of the same form as $X$, taking constant value on all edges or otherwise we would not have $\inflow X = \outflow X$. Since all edge flows on $C_4$ are automatically curl-free, $\ker (\Delta_1)=\ker (\dive)$ and is given by the set of all constant multiples of $X$. The number
\[
\beta_1(G) = \dim \ker (\Delta_1)
\]
counts the number of `$1$-dimensional holes' of $G$ and in this case we see that indeed $\beta_1(C_4) = 1$. To be a bit more precise, the `$1$-dimensional holes' are the regions that remain uncovered after the cliques are filled in.

We now turn our attention to the contrasting case of $C_3$. Looking at Figure~\ref{fig:eg1}, it may seem that $C_3$ also has a `$1$-dimensional hole' as in $C_4$ but this is a fallacy --- holes bounded by triangles are not regarded as holes in our framework.

For $C_3$ it is in fact true that every curl-free edge flow is a gradient. To see this, note that as in the case of $C_4$, any divergence-free $X \in  L^2_\wedge (E (C_3))$ must be constant on all edges and so 
\[
(\curl X)(1,2,3) = X(1,2) + X(2,3) + X(3,1) = c + c + c = 3c,
\]
for some $c \in \mathbb{R}$. If a divergence-free $X$ is also curl-free, then $c = 0$ and so $X = 0$. Hence for $C_3$, $\ker(\Delta_1) = \{0\}$ by \eqref{eq:bm} and $\ker(\curl) = \im(\grad)$ by \eqref{eq:helm2}. It also follows that $\beta_1(C_3) = 0$ and so $C_3$ has no `$1$-dimensional hole'.

What we have illustrated with $C_3$ and $C_4$ extends to any arbitrary graph. A moment's thought would reveal that the property $\beta_1(G) = 0$ is satisfied by any \textit{chordal graph}, i.e., one for which every cycle subgraph of four or more vertices has a \textit{chord}, an edge that connects two vertices of the cycle subgraph but that is not part of the cycle subgraph. Equivalently, a chordal graph is one where every chordless cycle subgraph is $C_3$.
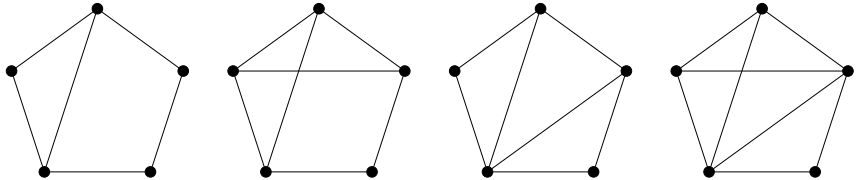
\begin{figure}[h]
\SetGraphUnit{3}
\SetVertexSimple[MinSize = 10pt]
\centering
\scalebox{0.4}{
\begin{tikzpicture}[rotate = 18]
  \Vertices{circle}{A,B,C,D,E}
  \Edges(A,B,C,D,E,A)
  \Edges(B,D)
\end{tikzpicture}}
\quad
\scalebox{0.4}{
\begin{tikzpicture}[rotate = 18]
  \Vertices{circle}{A,B,C,D,E}
  \Edges(A,B,C,D,E,A)
  \Edges(B,D)
  \Edges(A,C)
\end{tikzpicture}}
\quad
\scalebox{0.4}{
\begin{tikzpicture}[rotate = 18]
  \Vertices{circle}{A,B,C,D,E}
  \Edges(A,B,C,D,E,A)
  \Edges(B,D)
  \Edges(A,D)
\end{tikzpicture}}
\quad
\scalebox{0.4}{
\begin{tikzpicture}[rotate = 18]
  \Vertices{circle}{A,B,C,D,E}
  \Edges(A,B,C,D,E,A)
  \Edges(B,D)
  \Edges(C,A,D)
\end{tikzpicture}}
\caption{Left two graphs: not chordal. Right two graphs: chordal.}
\label{fig:eg3}
\end{figure}
\end{example}

\section{Higher order}\label{sec:HO}

We expect the case of alternating functions on edges, i.e., $k=1$, discussed in Section~\ref{sec:graph} to be the most useful in applications. However for completeness and since it is no more difficult to generalize to $k > 1$, we provide the analogue of  Section~\ref{sec:graph} for arbitrary $k$ here.

\subsection{Higher-order cochains}

Let $K(G)$ be the clique complex of a graph $G=(V,E)$ as defined in Section~\ref{sec:cliques}. We will write $K_k =K_k(G)$ for simplicity.

A $k$-\textit{cochain} (or $k$-form) is an alternating function on $K_{k+1}$, or more specifically, $f: V \times \dots \times V \rightarrow \mathbb{R}$ where
\begin{equation}\label{eq:altfn}
f(i_{\sigma(0)},\dots,i_{\sigma(k)})=\sgn (\sigma)f(i_{0},\dots,i_{k})
\end{equation}
for all $\{i_{0},\dots,i_{k}\}\in K_{k+1}$ and all  $\sigma \in\mathfrak{S}_{k+1}$, the symmetric group of permutations on $\{0,\dots,k\}$. We set $f(i_{0},\dots,i_{k}) = 0$ if $\{i_{0},\dots,i_{k}\}\not\in K_{k+1}$.

Again, we may put an inner product on $k$-cochains,
\[
\langle f,g\rangle =\sum\nolimits_{i_0 <\dots <i_k} w_{i_0\cdots i_k}f(i_{0},\dots,i_{k}) g(i_{0},\dots,i_{k}),
\]
with any positive weights satisfying $w_{i_{\sigma(0)} \cdots i_{\sigma(k)}} = w_{i_0\cdots i_k}$ for all $\sigma \in \mathfrak{S}_{k+1}$.

We denote the resulting Hilbert space by $L^2_\wedge(K_{k+1})$. This is a subspace of $L^2(\Wedge^{k+1} V)$, the space of alternating functions with $k+1$ arguments in $V$. Clearly,
\[
\dim L^2_\wedge(K_{k+1}) = \# K_{k+1}.
\]

A word of caution regarding the terminology: a $k$-cochain is a function on a $(k+1)$-clique and has $k+1$ arguments. The reason is due to the different naming conventions ---  a $(k+1)$-clique in graph theory is called a $k$-simplex in topology. In topological lingo, a vertex is a $0$-simplex, an edge a $1$-simplex, a triangle a $2$-simplex, a tetrahedron a $3$-simplex.

\subsection{Higher-order coboundary operators}\label{sec:HO2}

The $k$-\textit{coboundary operators} $\delta_{k}: L^2_\wedge(K_{k})\rightarrow L^2_\wedge(K_{k+1})$ are defined by
\begin{equation}\label{eq:cob}
(\delta_{k}f)(i_{0},\dots,i_{k+1})=\sum_{j=0}^{k+1}(-1)^{j}f(i_{0},\dots,i_{j-1},i_{j+1},\dots,i_{k+1}),
\end{equation}
for $k = 0,1,2,\dots.$ Readers familiar with differential forms may find it illuminating to think of coboundary operators as discrete analogues of exterior derivatives. Note that $f$ is a function with $k+1$ arguments but $\delta_k f $ is a function with $k+2$ arguments. A convenient and often-used notation is to put a carat over the omitted argument
\begin{equation}\label{eq:carat}
f(i_{0},\dots,\widehat{i}_{j},\dots,i_{k+1}) \coloneqq  f(i_{0},\dots,i_{j-1},i_{j+1},\dots,i_{k+1}).
\end{equation}

The crucial relation $AB=0$ in Section~\ref{sec:ped} is in fact
\begin{equation}\label{eq:fot2}
\delta_{k} \delta_{k-1} = 0,
\end{equation}
which  may be verified using \eqref{eq:cob} (see Theorem~\ref{thm:fund}). The equation \eqref{eq:fot2} is often verbalized as ``the coboundary of a coboundary is zero.'' It generalizes \eqref{eq:cg} and is sometimes called the \textit{fundamental theorem of topology}.

As in Section~\ref{sec:coho}, \eqref{eq:fot2} is equivalent to saying that $\im(\delta_{k-1}) $ is a subspace of $\ker(\delta_{k})$. We define the $k$th \textit{cohomology group} of $G$ to be the quotient vector space
\begin{equation}\label{eq:cohogp}
H^k(G) = \ker(\delta_{k})/\im(\delta_{k-1}),
\end{equation}
for $k =1,2,\dots, \omega(G) - 1$.

To keep track of the coboundary operators, it is customary to assemble them into a sequence of maps written in the form
\[
L^2_\wedge(K_0) \overset{\delta_0}{\longrightarrow}  L^2_\wedge(K_1) \overset{\delta_1}{\longrightarrow}\cdots\overset{\delta_{k-1}}{\longrightarrow}   L^2_\wedge(K_{k}) 
\overset{\delta_{k}}{\longrightarrow}  L^2_\wedge(K_{k+1})\overset{\delta_{k+1}}{\longrightarrow} \cdots \overset{\delta_{\omega}}{\longrightarrow}  L^2_\wedge(K_{\omega}).
\]
This sequence is called a \textit{cochain complex}. It is said to be \textit{exact} if $\im(\delta_{k-1}) = \ker(\delta_{k})$ or, equivalently, $H^k(G)  = \{0\}$, for all $k =1,2,\dots,\omega(G)-1$.

For $k = 1$, we get $\delta_0= \grad$, $\delta_1 = \curl$, and the first two terms of the cochain complex are
\[
 L^2(V) \xrightarrow{\grad}  L^2_\wedge(E) \xrightarrow{\curl}  L^2_\wedge(T).
\]

\subsection{Hodge theory}\label{sec:HO3}

The \textit{Hodge $k$-Laplacian} $\Delta_{k}:L^2_\wedge(K_{k}) \to L^2_\wedge(K_{k})$ is defined as
\[
\Delta_{k} = \delta_{k-1} \delta_{k-1}^{\ast}+ \delta_{k}^{\ast}\delta_{k}.
\]
We call $f \in L^2_\wedge(K_{k}) $ a \textit{harmonic $k$-cochain} if it satisfies the Laplace equation
\[
\Delta_k f = 0.
\]

Applying the results in Section~\ref{sec:hodge} with $A = \delta_{k}$ and $B = \delta_{k-1}$, we obtain the unique representation of cohomology classes as harmonic cochains
\[
H^k(G) =\ker(\delta_k)/\im(\delta_{k-1}) \cong \ker(\delta_{k})\cap\ker(\delta_{k-1}^{\ast}) = \ker(\Delta_k),
\]
as well as the Hodge decomposition
\begin{equation}\label{eq:hohd}
L^2_\wedge(K_{k}) = \rlap{$\overbrace{\phantom{\im(\delta_{k}^*) \oplus \ker(\Delta_k)}}^{\ker(\delta_{k-1}^*)}$}\im(\delta_{k}^*) \oplus \underbrace{\ker(\Delta_k) \oplus \im(\delta_{k-1})}_{\ker(\delta_{k})},
\end{equation}
and the relation
\[
\im(\Delta_k) = \im(\delta_k^*) \oplus \im(\delta_{k-1}).
\]

\begin{example}[Hearing the shape of a graph]\label{eg:gip}
Two undirected graphs $G$ and $H$ on $n$ vertices are said to be \textit{isomorphic} if they are essentially the same graph up to relabeling of vertices. The \textit{graph isomorphism problem}, an open problem in computer science, asks whether there is a polynomial-time algorithm\footnote{An astounding recent result of Babai \cite{babai2} is that there is a \emph{quasipolynomial}-time algorithm.}  for deciding if two given graphs are isomorphic \cite{babai}.  Clearly two isomorphic graphs must be \textit{isospectral} in the sense that the  eigenvalues (ordered and counted with multiplicities) of their graph Laplacians are equal,
\[
\lambda_i(\Delta_0(G)) = \lambda_i(\Delta_0(H)), \quad i =1,\dots,n,
\]
a condition that can be checked in polynomial time. Not surprisingly, the converse --- the graph theoretic analogue of Kac's famous problem \cite{Kac} --- is not true, or we would have been able to determine graph isomorphism in polynomial time. We should mention that there are several definitions of isospectral graphs, in terms of the adjacency matrix, graph Laplacian, normalized Laplacian, signless Laplacian, etc; see \cite{BG,GLS} for many interesting examples of nonisomorphic isospectral graphs.

The reader may perhaps wonder what happens if we impose the stronger requirement that the eigenvalues of all their higher-order Hodge $k$-Laplacians be  equal as well?
\[
\lambda_i(\Delta_k(G)) = \lambda_i(\Delta_k(H)), \quad i =1,\dots,n, \; k =0,\dots,m.
\]
For any $m\ge 1$, these indeed give a stronger set of sufficient conditions that can be checked in polynomial time. For example, the eigenvalues of $\Delta_0$ for the two graphs in Figure~\ref{fig:iso1} are $0, 0.76, 2, 3, 3, 5.24$ (all numbers rounded to two decimal figures). On the other hand, the eigenvalues of $\Delta_1$ are $0, 0.76, 2, 3, 3, 3, 5.24$ for the graph on the left and $0, 0, 0.76, 2, 3, 3, 5.24$ for the graph on the right, allowing us to conclude that they are not isomorphic. These calculations are included in Section~\ref{sec:cal}.
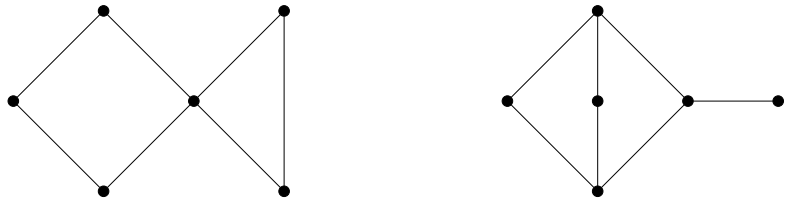
\begin{figure}[h]
\centering
\SetVertexSimple[MinSize = 10pt]
\scalebox{0.4}{
\begin{tikzpicture}
\SetGraphUnit{3}
\coordinate (O) at (0,0);
\WE(O){1} \NO(O){2} \EA(O){3} \SO(O){4}
\NOEA(3){5} \SOEA(3){6}
\Edges(1,2,3,5,6,3,4,1)
\end{tikzpicture}}
\hspace*{1in}
\scalebox{0.4}{
\begin{tikzpicture}
\SetGraphUnit{3}
\Vertex{6}
\WE(6){1} \NO(6){2} \EA(6){3} \SO(6){4}
\EA(3){5}
\Edges(4,1,2,3,5)
\Edges(3,4,6,2)
\end{tikzpicture}}
\caption{These graphs have isospectral Laplacians (Hodge $0$-Laplacians) but not Helmholtzians (Hodge $1$-Laplacians).}
\label{fig:iso1}
\end{figure}

Non-isomorphic graphs can nevertheless have isospectral Hodge Laplacians of all order. The two graphs in Figure~\ref{fig:iso2} are clearly non-isomorphic. Neither contains cliques of order higher than two, so their Hodge $k$-Laplacians are zero for all $k > 2$. We may check (see Section~\ref{sec:cal}) that  the first three Hodge Laplacians $\Delta_0$, $\Delta_1$, $\Delta_2$, of both graphs are isospectral.
\begin{figure}[h]
\centering
\SetVertexSimple[MinSize = 10pt]
\scalebox{0.4}{
\begin{tikzpicture}
\SetGraphUnit{3}
\Vertex{2}
\NO(2){1} \EA(2){3}
\EA(3){4}
\NOEA(4){5} \EA(4){6} \SOEA(4){7}
\Edges(2,1,3,2)
\Edges(3,4,6)
\Edges(4,5)
\Edges(4,7)
\end{tikzpicture}}
\hspace*{1in}
\scalebox{0.4}{
\begin{tikzpicture}
\SetGraphUnit{3}
\Vertex{2}
\NO(2){1} \EA(2){3}
\EA(3){4} \SOWE(3){6}
\NOEA(4){5} \SOEA(4){7}
\Edges(2,1,3,2)
\Edges(3,4)
\Edges(3,6)
\Edges(4,5)
\Edges(4,7)
\end{tikzpicture}}
\caption{Non-isormorphic graphs with isospectral Hodge $k$-Laplacians for all $k = 0, 1, 2,\dots.$}
\label{fig:iso2}
\end{figure}
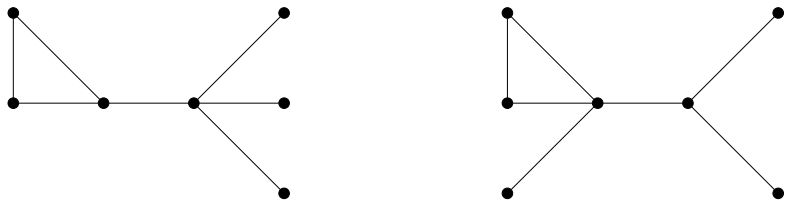

\end{example}

\begin{table}[h!]
\centering
\caption{Topological jargons (second pass)}
\label{tab:terms2}
\vspace*{-1.5ex}
\renewcommand{\arraystretch}{1.25}
\begin{tabular}{|l|l|}
\hline
\textsc{name} & \textsc{meaning}\\
\hline
coboundary maps & $\delta_{k}: L^2_\wedge(K_{k})\rightarrow L^2_\wedge(K_{k+1})$\\
cochains & elements $f \in L^2_\wedge(K_{k})$\\
cochain complex & $\cdots \longrightarrow L^2_\wedge(K_{k-1})  \overset{\delta_{k-1}}{\longrightarrow}  L^2_\wedge(K_{k}) \overset{\delta_{k}}{\longrightarrow}  L^2_\wedge(K_{k+1}) \longrightarrow  \cdots$\\
cocycles & elements of $\ker(\delta_{k})$\\
coboundaries & elements of $\im(\delta_{k-1})$\\
cohomology classes & elements of $ \ker(\delta_{k})/\im(\delta_{k-1})$\\
harmonic cochains & elements of $\ker(\Delta_{k} ) $\\
Betti numbers & $\dim \ker(\Delta_{k} ) $\\
Hodge Laplacians & $\Delta_{k} =\delta_{k-1} \delta_{k-1}^{\ast}+ \delta_{k}^{\ast}\delta_{k}$\\
$f$ is closed & $\delta_{k}f = 0$\\
$f$ is exact & $f = \delta_{k-1}g$ for some $g \in L^2_\wedge(K_{k-1})$\\
$f$ is coclosed & $\delta_{k-1}^*f = 0$\\
$f$ is coexact & $f = \delta_{k}^*h$ for some $h \in L^2_\wedge(K_{k+1})$\\
$f$ is harmonic & $\Delta_k f =0 $\\
\hline
\end{tabular}
\end{table}

\section{Detailed proofs and calculations}\label{sec:proofs}
\allowdisplaybreaks

In this section, we  provide  proofs of the linear algebraic facts in Section~\ref{sec:ped},  verify  various claims in Sections~\ref{sec:graph} and \ref{sec:HO}, and work out the details of Example~\ref{eg:gip}. 

\subsection{Linear Algebra over $\mathbb{R}$}\label{sec:NLA}

We provide routine proofs for some linear algebraic facts that we have used freely in Section~\ref{sec:ped}. We will work over $\mathbb{R}$ for convenience but every statement in Theorems~\ref{thm:fredholm}, \ref{thm:hodge}, \ref{thm:cohomology} extends to any subfield of $\mathbb{C}$.

\begin{theorem}\label{thm:fredholm}
Let $A \in \mathbb{R}^{m \times n}$. Then
\begin{dingautolist}{192}
\item $\ker(A^* A) = \ker(A)$,
\item $\im(A^* A) = \im(A^*)$,
\item $\ker(A^*) = \im(A)^\perp$,
\item $\im(A^*) = \ker(A)^\perp$,
\item $\mathbb{R}^n = \ker(A) \oplus \im(A^*)$.
\end{dingautolist}
\end{theorem}
\begin{proof}\hfill
\begin{dingautolist}{192}
\item Clearly $\ker(A) \subseteq \ker(A^* A)$. If $A^*Ax = 0$, then $\lVert Ax \rVert^2 = x^*A^*Ax = 0$, so $Ax = 0$, and so $\ker(A^* A) \subseteq \ker(A)$.
\item Applying rank-nullity theorem twice with  \ding{192}, we get
\begin{align*}
\rank (A^*A ) &= n - \nullity(A^*A) \\
&= n - \nullity(A) = \rank(A) = \rank(A^*).
\end{align*}
Since $\im(A^*A) \subseteq \im(A^* )$, the result follows.
\item If $x \in  \im(A)^\perp$, then $0 =\langle x, Ay \rangle = \langle A^*x, y\rangle$ for all $y \in \mathbb{R}^n$, so $A^*x = 0$. If $x \in \ker(A^*)$, then $\langle x, Ay \rangle = \langle A^*x, y\rangle = 0$ for all $y \in \mathbb{R}^n$, so $x \in \im(A)^\perp$.
\item By \ding{194}, $\im(A^*)^\perp = \ker(A^{**}) =  \ker(A)$ and result follows.
\item $\mathbb{R}^n = \ker(A) \oplus \ker(A)^\perp = \ker(A) \oplus \im(A^*)$ by \ding{195}.
\end{dingautolist}
\end{proof}

Our next proof ought to convince readers that the Hodge decomposition theorem \ding{200} is indeed an extension of the Fredholm alternative theorem  \ding{196} to a pair of matrices.
\begin{theorem}\label{thm:hodge}
Let $A \in \mathbb{R}^{m \times n}$ and $B \in \mathbb{R}^{n \times p}$ with $AB = 0$. Then
\begin{dingautolist}{197}
\item $\ker(A^*A + BB^*) = \ker(A) \cap \ker(B^*)$,
\item $\ker(A) =  \im(B) \oplus \ker(A^*A + BB^*) $,
\item $\ker(B^*) =  \im(A^*)  \oplus \ker(A^*A + BB^*) $,
\item $\mathbb{R}^n = \im(A^*) \oplus \ker(A^*A + BB^*) \oplus \im(B)$,
\item $\im(A^*A + BB^*)  = \im(A^*) \oplus \im(B)$.
\end{dingautolist}

\end{theorem}
\begin{proof} Note that $\im(B) \subseteq \ker(A)$ as $AB = 0$, $\im(A^*) \subseteq \ker(B^*)$ as $B^*A^* = 0$.
\begin{dingautolist}{197}
\item Clearly $\ker(A) \cap \ker(B^*) \subseteq \ker(A^*A + BB^*)$. Let $x \in  \ker(A^*A + BB^*)$. Then $A^*A x = -BB^*x$.
\begin{itemize}
\item Multiplying by $A$, we get $AA^*A x = -ABB^*x = 0$ since $AB = 0$. So $A^*A x \in \ker(A)$. But $A^*A x \in \im(A^*) = \ker(A)^\perp$  by \ding{195}. So $A^*A x = 0$ and $x \in \ker(A^*A) =\ker(A)$ by \ding{192}.
\item Multiplying by $B^*$, we get $0 = B^*A^*A x = -B^*BB^*x$ since $B^*A^* = 0$. So $BB^* x \in \ker(B^*)$. But $BB^* x \in \im(B) = \ker(B^*)^\perp$  by \ding{194}. So $BB^* x = 0$ and $x \in \ker(BB^*) =\ker(B^*)$ by \ding{192}.
\end{itemize}
Hence $x \in \ker(A) \cap \ker(B^*)$.

\item  Applying \ding{196} to $B^*$,
\begin{align*}
\ker(A) &= \mathbb{R}^n \cap \ker(A) = [\ker(B^*) \oplus \im(B)] \cap \ker(A) \\
&=  [\ker(B^*) \cap \ker(A)] \oplus [\im(B) \cap \ker(A)  ] \\
& =  \ker(A^*A + BB^*)  \oplus \im(B),
\end{align*}
where the last equality follows from \ding{197} and $\im(B) \subseteq \ker(A)$.

\item Applying \ding{196},
\begin{align*}
\ker(B^*) &= \mathbb{R}^n \cap \ker(B^*) = [\ker(A) \oplus \im(A^*)] \cap \ker(B^*) \\
&=  [\ker(A) \cap \ker(B^*)] \oplus [\im(A^*) \cap \ker(B^*)  ] \\
& =  \ker(A^*A + BB^*)  \oplus \im(A^*),
\end{align*}
where the last equality follows from \ding{197} and $\im(A^*) \subseteq \ker(B^*)$. Alternatively, apply  \ding{198} with $B^*,A^*$ in place of $A,B$.

\item Applying \ding{196} to $B^*$ followed by \ding{199}, we get
\[
\mathbb{R}^n = \ker(B^*) \oplus \im(B) =    \im(A^*) \oplus \ker(A^*A + BB^*) \oplus  \im(B). 
\]

\item Applying \ding{196} to $A^*A + BB^*$, which is self-adjoint, we see that
\[
\im(A^*A + BB^*) = \ker(A^*A + BB^*)^\perp = \im(A^*) \oplus \im(B),
\]
where the last equality follows from \ding{200}. 
\end{dingautolist}
\end{proof}

Any two vector spaces of the same dimension are isomorphic. So saying that two vector spaces are isomorphic isn't saying very much --- just that they have the same dimension. The two spaces in \eqref{eq:cohomology} are special because they are \textit{naturally isomorphic}, i.e., if you construct an isomorphism, and the guy in the office next door constructs an isomorphism, both of you would end up with the same isomorphism, namely, the one below.
\begin{theorem}\label{thm:cohomology}
Let $A \in \mathbb{R}^{m \times n}$ and $B \in \mathbb{R}^{n \times p}$ with $AB = 0$. Then the following spaces are naturally isomorphic
\[
\ker(A)/\im(B) \cong \ker(A) \cap \ker(B^*) \cong \ker(B^*)/\im(A^*).
\]
\end{theorem}
\begin{proof}
Let $\pi : \mathbb{R}^n \to \im(B)^\perp$ be the orthogonal projection of $\mathbb{R}^n$ onto the orthogonal complement of $\im(B)$. So any $x \in \mathbb{R}^n$ has a unique decomposition into two mutually orthogonal components
\[
\setlength{\arraycolsep}{1.5pt}
\begin{matrix}
\mathbb{R}^n &=& \im(B)^\perp &\oplus & \im(B),\\[0.5ex]
x &=& \pi(x) &+& (1 -\pi)(x) .
\end{matrix}
\]
Let $\pi_A$ be $\pi$ restricted to the subspace $\ker(A)$. So any $x \in \ker(A) $ has a unique decomposition into two mutually orthogonal components
\[
\setlength{\arraycolsep}{1.5pt}
\begin{matrix}
\ker(A) &=& \bigl( \ker(A) \cap\im(B)^\perp \bigr) &\oplus& \im(B),\\[0.5ex]
x &=& \pi_A(x) &+& (1 -\pi_A)(x),
\end{matrix}
\]
bearing in mind that $\ker(A) \cap\im(B) = \im(B)$ since $\im(B) \subseteq \ker(A)$.

As $\pi$ is surjective, so is $\pi_A$. Hence $\im(\pi_A) =   \ker(A) \cap  \im(B)^\perp $. Also, for any $x \in \ker(A)$, $\pi_A (x) = 0$ iff the component of $x$ in $\im(B)^\perp$ is zero, i.e., $x \in \im(B)$. Hence  $\ker(\pi_A)  = \im(B)$. The first isomorphsim theorem,
\[
\ker(A)/\ker(\pi_A) \cong \im(\pi_A) =   \ker(A) \cap  \im(B)^\perp
\]
yields the required result since $\im(B)^\perp = \ker(B^*)$ by \ding{194}. The other isomorphism may be obtained as usual by using $B^*,A^*$ in place of $A,B$.
\end{proof}

In mathematics, \textit{linear algebra} usually refers to a collection of facts that follow from the defining axioms of a field and of a vector space. In this regard, every single statement in Theorems~\ref{thm:fredholm}, \ref{thm:hodge}, \ref{thm:cohomology} is false as a statement in linear algebra --- they depend specifically on our working over a subfield of $\mathbb{C}$ and are not true over arbitrary fields. For example, consider the finite field of two elements $\mathbb{F}_2 = \{0,1\}$ and take
\[
A = B = \begin{bmatrix}
1 & 1\\
1 & 1
\end{bmatrix} .
\]
Then $A^* = A = B = B^*$, and $AB = B^*A^* = A^*A = BB^* =A^*A + BB^* = 0$, which serves as a counterexample to \ding{192}, \ding{193}, \ding{196}, \ding{198}, \ding{199}, \ding{200}, \ding{201}, and Theorem~\ref{thm:cohomology}.

\subsection{Div, Grad, Curl, and All That}\label{sec:RV}

We provide routine verifications of statements claimed in Sections~\ref{sec:graph} and \ref{sec:HO}.

\begin{lemma}\label{lem:grad*}
Equip $L^2(V)$ and $L^2_\wedge(E)$ with the inner products in \eqref{eq:inner}, we have
\[
\grad^* X(i) = -\sum_{j=1}^n \frac{w_{ij}}{w_i} X(i,j) = -\dive X(i).
\]
\end{lemma}
\begin{proof}
The required expression follows from
\begin{align*}
\langle \grad^* X, f \rangle_V &= \langle X, \grad f \rangle_E\\
&= \sum\nolimits_{i < j} w_{ij} X(i,j) \grad f(i,j)\\
&= \sum\nolimits_{i < j} w_{ij} X(i,j) [f(j) - f(i)]\\
&=\sum\nolimits_{i < j} w_{ij} X(i,j) f(j) +\sum\nolimits_{i < j} w_{ij} X(j,i)  f(i)\\
&\stackrel{\text{\ding{192}}}{=}\sum\nolimits_{j < i} w_{ji} X(j,i) f(i) +\sum\nolimits_{i < j} w_{ij} X(j,i)  f(i)\\
&\stackrel{\text{\ding{193}}}{=}\sum\nolimits_{j < i} w_{ij} X(j,i) f(i) +\sum\nolimits_{i < j} w_{ij} X(j,i)  f(i)\\
&=\sum\nolimits_{i \ne j} w_{ij} X(j,i) f(i) \\
&=\sum\nolimits_{ i =1}^n w_i \Bigl[\sum\nolimits_{j: j \ne i } \frac{w_{ij}}{w_i} X(j,i) \Bigr] f(i)\\
&\stackrel{\text{\ding{194}}}{=}\sum\nolimits_{ i =1}^n w_i \Bigl[\underbrace{\sum\nolimits_{j=1}^n \frac{w_{ij}}{w_i} X(j,i) }_{\grad^* X(i)}\Bigr] f(i).
\end{align*}
\begin{dingautolist}{192}
\item follows from swapping labels $i$ and $j$ in the first summand.
\item follows from $w_{ij} = w_{ji}$. 
\item follows from $X(i,i) =0$.
\end{dingautolist}
\end{proof}

\begin{lemma}\label{lem:curl*}
Equip $L^2_\wedge(E)$ and $L^2_\wedge(T)$ with the inner products in \eqref{eq:inner}, we have
\[
\curl^* \Phi(i,j) =\sum_{k=1}^n \frac{w_{ijk}}{w_{ij}} \Phi(i,j,k).
\]
\end{lemma}
\begin{proof}
The required expression follows from
\begin{align*}
\langle \curl^* \Phi, X \rangle_E &= \langle \Phi, \curl X \rangle_T = \sum\nolimits_{i < j < k} w_{ijk} \Phi(i,j,k) \curl X(i,j,k)\\
&= \sum\nolimits_{i < j < k} w_{ijk} \Phi(i,j,k) [X(i,j) + X(j,k)+X(k,i)]\\
&=\sum\nolimits_{i < j < k} w_{ijk} \Phi(i,j,k) X(i,j) +\sum\nolimits_{i < j < k} w_{ijk} \Phi(i,j,k)  X(j,k)\\
&\qquad \qquad  +\sum\nolimits_{i < j < k} w_{ijk} \Phi(i,j,k)  X(k,i)\\
&\stackrel{\text{\ding{192}}}{=}\sum\nolimits_{i < j < k} w_{ijk} \Phi(i,j,k) X(i,j) +\sum\nolimits_{i < j < k} w_{ijk} \Phi(j,k,i)  X(j,k)\\
&\qquad \qquad  +\sum\nolimits_{i < j < k} w_{ijk} \Phi(k,i,j)  X(k,i)\\
&\stackrel{\text{\ding{193}}}{=}\sum\nolimits_{i < j < k} w_{ijk} \Phi(i,j,k) X(i,j) +\sum\nolimits_{k < i < j} w_{kij} \Phi(i,j,k) X(i,j)\\
&\qquad \qquad  +\sum\nolimits_{i< k < j} w_{ikj} \Phi(j,i,k) X(j,i)\\
&\stackrel{\text{\ding{194}}}{=}\sum\nolimits_{i < j < k} w_{ijk} \Phi(i,j,k) X(i,j) +\sum\nolimits_{k < i < j} w_{kij} \Phi(i,j,k) X(i,j)\\
&\qquad \qquad  +\sum\nolimits_{i < k < j} w_{ikj} \Phi(i,j,k) X(i,j)\\
&=\sum\nolimits_{i < j}\Bigl[ \Bigl( \sum\nolimits_{k=j+1}^n + \sum\nolimits_{k=1}^{i-1} + \sum\nolimits_{k=i+1}^{j-1} \Bigr) w_{ijk} \Phi(i,j,k)\Bigr] X(i,j)\\
&=\sum\nolimits_{i < j}w_{ij}\Bigl[ \sum\nolimits_{k: k \ne i,j} \frac{w_{ijk}}{w_{ij}} \Phi(i,j,k)\Bigr] X(i,j)\\
&\stackrel{\text{\ding{195}}}{=}\sum\nolimits_{i < j}w_{ij}\Bigl[ \underbrace{\sum\nolimits_{k=1}^n \frac{w_{ijk}}{w_{ij}} \Phi(i,j,k)}_{\curl^* \Phi(i,j)}\Bigr] X(i,j).
\end{align*}
\begin{dingautolist}{192}
\item follows from the alternating property of $\Phi$.
\item follows from relabeling $j,k,i$ as $i,j,k$ in the second summand and swapping labels $j$ and $k$ in the third summand.
\item follows from $\Phi(j,i,k) X(j,i) = \Phi(i,j,k) X(i,j) $ since both changed signs.
\item follows from $\Phi(i,j,i) = \Phi(i,j,j) = 0$.
\end{dingautolist}
\end{proof}

\begin{lemma}\label{lem:gl}
The operator $\Delta_0 = -\dive  \grad$ gives us the usual graph Laplacian.
\end{lemma}
\begin{proof}
Let $f \in L^2(V)$. By definition,
\[
\grad f(i,j) =
\begin{cases}
f(j) - f(i) &\text{if } \{i,j\} \in E,\\
0 &\text{otherwise}.
\end{cases}
\]
Define the adjacency matrix $A \in \mathbb{R}^{n \times n}$ by
\[
a_{ij} =
\begin{cases}
1 & \text{if } \{i, j\} \in E,\\
0 & \text{otherwise}.
\end{cases}
\]
The gradient may be written as $\grad f(i,j) = a_{ij} (f(j) - f(i))$ and so
\begin{equation}\label{eq:familiar}
\begin{aligned}
(\Delta_0 f)(i) &= - [\dive (\grad f)] (i)  = -[\dive a_{ij}(f (j) - f(i))](i) \\
&= - \sum\nolimits_{j=1}^n a_{ij}[ f(j) - f(i) ] = d_i f(i) - \sum\nolimits_{j=1}^n  a_{ij} f(j),
\end{aligned}
\end{equation}
where for any vertex $i =1,\dots,n$, we define its degree as
\[
d_i = \deg(i) = \sum\nolimits_{j=1}^n a_{ij}.
\]
If we regard a function $f \in L^2(V)$ as a vector $(f_1,\dots,f_n) \in \mathbb{R}^n$ where $f(i) = f_i$ and set $D = \diag(d_1,\dots, d_n)  \in \mathbb{R}^{n \times n}$, then \eqref{eq:familiar} becomes
\[
\Delta_0 f =
\begin{bmatrix}
d_1 - a_{11} & - a_{12} &\cdots & -a_{1n}\\
-a_{21} & d_2 - a_{22} &\cdots & - a_{2n}\\
\vdots &   &\ddots &  \vdots\\
-a_{n1} & - a_{n2} &\cdots & d_n-a_{nn}
\end{bmatrix}
\begin{bmatrix}
f_1\\
f_2\\
\vdots\\
f_n
\end{bmatrix} = (D - A)f.
\]
So $\Delta_0$ may be regarded as  $D - A$, the usual definition of a graph Laplacian.
\end{proof}

\begin{theorem}\label{thm:fund}
We have that
\[
\curl \grad  = 0, \qquad \dive \curl^* = 0,
\]
and more generally, for $k = 1,2,\dots,$
\[
\delta_k \delta_{k-1} = 0, \qquad \delta_{k-1}^* \delta_k^* = 0.
\]
\end{theorem}
\begin{proof}
We only need to check $\delta_k \delta_{k-1} = 0$. The other relations follow from taking adjoint or specializing to $k=1$. Let $f \in L^2_\wedge(K_{k-1})$. By \eqref{eq:cob} and \eqref{eq:carat},
\begin{align*}
&(\delta_{k} \delta_{k-1} f)(i_{0},\dots,i_{k+1}) =\sum\nolimits_{j=0}^{k+1}(-1)^{j} \delta_{k-1}f(i_{0},\dots,\widehat{i}_{j},\dots,i_{k+1})\\
&\qquad \stackrel{\text{\ding{192}}}{=}\sum\nolimits_{j=0}^{k+1}(-1)^{j} \Bigl[ \sum\nolimits_{\ell =0}^{j-1}(-1)^{\ell}  f(i_{0},\dots,\widehat{i}_{\ell},\dots,\widehat{i}_{j},\dots,i_{k+1})\\
&\qquad\qquad\qquad\qquad\qquad\qquad +\sum\nolimits_{\ell =j+1}^{k+1}(-1)^{\ell-1}  f(i_{0},\dots,\widehat{i}_{j},\dots,\widehat{i}_{\ell},\dots,i_{k+1})\Bigr]\\
&\qquad =\sum\nolimits_{j < \ell }(-1)^{j} (-1)^{\ell}  f(i_{0},\dots,\widehat{i}_{j},\dots,\widehat{i}_{\ell},\dots,i_{k+1})\\
&\qquad\qquad\qquad+ \sum\nolimits_{j > \ell }(-1)^{j} (-1)^{\ell-1}  f(i_{0},\dots,\widehat{i}_{\ell},\dots,\widehat{i}_{j},\dots,i_{k+1})\\
&\qquad \stackrel{\text{\ding{193}}}{=}\sum\nolimits_{j < \ell }(-1)^{j+\ell}  f(i_{0},\dots,\widehat{i}_{j},\dots,\widehat{i}_{\ell},\dots,i_{k+1})\\
&\qquad\qquad\qquad + \sum\nolimits_{\ell > j} (-1)^{j + \ell-1}  f(i_{0},\dots,\widehat{i}_{j},\dots,\widehat{i}_{\ell},\dots,i_{k+1})\\
&\qquad =\sum\nolimits_{j < \ell }(-1)^{j+\ell}  f(i_{0},\dots,\widehat{i}_{j},\dots,\widehat{i}_{\ell},\dots,i_{k+1})\\
&\qquad\qquad\qquad - \sum\nolimits_{j < \ell} (-1)^{j + \ell}  f(i_{0},\dots,\widehat{i}_{j},\dots,\widehat{i}_{\ell},\dots,i_{k+1}) = 0.
\end{align*}
The power of $-1$ in the third sum in \ding{192} is $\ell - 1$ because an argument preceding $\widehat{i}_{\ell}$ is omitted and so $\widehat{i}_{\ell}$ is the $(\ell -1)$th argument (which is also omitted). \ding{193} follows from swapping labels $j$ and $\ell$ in the second sum.
\end{proof}

\subsection{Calculations}\label{sec:cal}

We will work out the details of Example~\ref{eg:gip}. While we have defined coboundary operators and Hodge Laplacians as abstract, coordinate-free linear operators, any actual applications would invariably involve `writing them down' as matrices to facilitate calculations. Readers might perhaps find our concrete approach here instructive.

A simple recipe for writing down a matrix representing a coboundary operator or a Hodge Laplacian is as follows: Given an undirected graph, label its vertices and edges arbitrarily but differently for easy distinction (e.g., we used numbers for vertices and letters for edges) and assign arbitrary directions to the edges. From the graphs in Figure~\ref{fig:iso1}, we get the labeled directed graphs $G_1$ (left) and $G_2$ (right) in Figure~\ref{fig:iso1a}.
\begin{figure}[h]
\centering
\scalebox{0.75}{\begin{tikzpicture}
\SetGraphUnit{2}
\coordinate (O) at (0,0);
\WE(O){1} \NO(O){2} \EA(O){3} \SO(O){4}
\NOEA(3){5} \SOEA(3){6}
\tikzset{EdgeStyle/.style = {->}}
\Edge[label=$a$](1)(2)
\Edge[label=$b$](2)(3)
\Edge[label=$c$](3)(4)
\Edge[label=$d$](4)(1)
\Edge[label=$e$](3)(5)
\Edge[label=$f$](5)(6)
\Edge[label=$g$](3)(6)
\end{tikzpicture}}
\qquad\qquad
\scalebox{0.75}{
\begin{tikzpicture}
\SetGraphUnit{2}
\Vertex{6}
\WE(6){1} \NO(6){2} \EA(6){3} \SO(6){4}
\EA(3){5}
\tikzset{EdgeStyle/.style = {->}}
\Edge[label=$a$](1)(2)
\Edge[label=$b$](2)(3)
\Edge[label=$c$](3)(4)
\Edge[label=$d$](4)(1)
\Edge[label=$e$](3)(5)
\Edge[label=$f$](4)(6)
\Edge[label=$g$](6)(2)
\end{tikzpicture}}
\caption{The graphs in Figure~\ref{fig:iso1}, with vertices and edges arbitrarily labeled and directions on edges arbitrarily assigned.}
\label{fig:iso1a}
\end{figure}
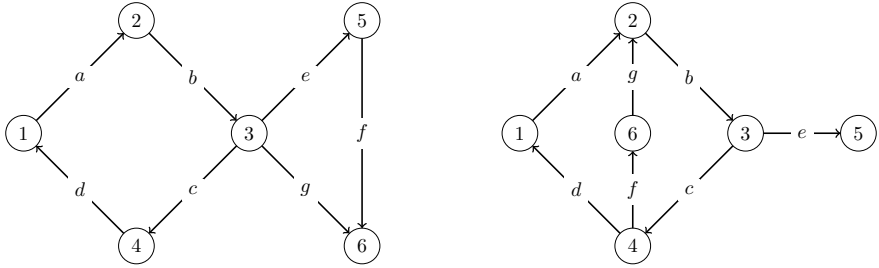

The next step is to write down a matrix whose columns are indexed by the vertices and the rows are indexed by the edges and whose $(i,j)$th entry is $+1$ if $j$th edge points into the $i$th vertex, $-1$ if $j$th edge points out of the $i$th vertex, and $0$ otherwise. This matrix represents the gradient operator $\delta_0 = \grad$. We get
\[
A_1 = \kbordermatrix{ 
& 1 & 2 &3 & 4 &5 &6 \\
a & -1 & 1 & 0 & 0 & 0 & 0\\
b & 0 & -1 & 1 & 0 & 0 & 0\\
c & 0 & 0 & -1 & 1 & 0 & 0\\
d & 1 & 0 & 0 & -1 & 0 & 0\\
e & 0 & 0 & -1 & 0 & 1 & 0\\
f & 0 & 0 & 0 & 0 & -1 & 1\\
g & 0 & 0 & -1 & 0 & 0 & 1
},
A_2 = \kbordermatrix{ 
& 1 & 2 &3 & 4 &5 &6 \\
a & -1    & 1     & 0     & 0     & 0     & 0 \\
b  & 0    & -1    & 1     & 0     & 0     & 0 \\
c  & 0     & 0    & -1    & 1     & 0     & 0 \\
d  & 1     & 0     & 0    & -1     & 0     & 0 \\
e  & 0     & 0    & -1     & 0    & 1     & 0 \\
f  & 0     & 0     & 0    & -1     & 0    & 1 \\
g  & 0    & 1     & 0     & 0     & 0    & -1
}
\]
for $G_1$ and $G_2$ respectively. Note that every row must contain exactly one $+1$ and one $-1$ since every edge is defined by a pair of vertices. This matrix is also known as a vertex-edge incidence matrix of the graph. Our choice of $\pm 1$ for in/out-pointing edges is also arbitrary --- the opposite choice works just as well as long as we are consistent throughout.

The graph Laplacians may either be computed from our definition as
\begin{align*}
L_1 = A_1^* A_1  &= \kbordermatrix{ 
& 1 & 2 &3 & 4 &5 &6 \\
1 & 2   & -1    & 0   & -1    & 0    & 0\\
2 & -1    & 2   & -1    & 0    & 0    & 0\\
3 & 0   & -1   & 4   & -1   & -1   & -1\\
4 & -1    & 0   & -1    & 2    & 0    & 0\\
5 & 0    & 0   & -1    & 0    & 2   & -1\\
6 & 0    & 0   & -1    & 0   & -1    & 2
},
\\
L_2 = A_2^* A_2 &= \kbordermatrix{ 
& 1 & 2 &3 & 4 &5 &6 \\
1 & 2   & -1    & 0   & -1    & 0    & 0\\
2 & -1    & 3   & -1    & 0    & 0   & -1\\
3 & 0   & -1    & 3   & -1   & -1    & 0\\
4 & -1    & 0   & -1    & 3    & 0   & -1\\
5 & 0    & 0   & -1    & 0    & 1    & 0\\
6 & 0   & -1    & 0   & -1    & 0    & 2
},
\end{align*}
or written down directly using the usual definition \cite{Chung, Spielman},
\[
\ell_{ij} =
\begin{cases}
\deg(v_i) & \text{if}\ i = j, \\
-1 & \text{if } v_i \text{ is adjacent to } v_j, \\
0 & \text{otherwise.}
\end{cases}
\]
We obtain the same Laplacian matrix irrespective of the choice of directions on edges and the choice of $\pm 1$ for in/out-pointing edges. For us there is no avoiding the gradient operators since we need them for the graph Helmholtzian below. 

We may now find the eigenvalues of $L_1$ and $L_2$ and see that they are indeed the values we claimed in Example~\ref{eg:gip}:
\[
\lambda(L_1) = (
   0, \;
    3-\sqrt{5}, \;
    2, \;
    3, \;
    3, \;
    3+\sqrt{5}  )
=
\lambda(L_2).
\]

To write down the graph Helmholtzians, we first observe that $G_1$ has exactly one triangle (i.e., $2$-clique) whereas $G_2$ has none\footnote{Those who see two triangles should note that these are really squares, or $C_4$'s to be accurate. See also Example~\ref{eg:curl}.}. We will need to label and pick an arbitrary orientation for the triangle in $G_1$: We denote it as $T$ and orient it clockwise $3 \to 5 \to 6 \to 3$. A matrix representing the operator $\delta_1 = \curl$ may be similarly written down by indexing the columns with edges and the rows with triangles. Here we make the arbitrary choice that if the $j$th edge points in the same direction as the orientation of the $i$th triangle, then the $(i,j)$th entry is $+1$ and if it points in the opposite direction, then the entry is $-1$. For $G_1$ we get
\[
B_1 = \kbordermatrix{ 
& a & b & c & d & d & e & f\\
T & 0    & 0    & 0    & 0    & 1    & 1   & -1
}.
\]
Since $G_2$ contains no triangles, $B_2 = 0$ by definition.

We compute the graph Helmholtzians from definition,
\begin{align*}
H_1 = A_1 A_1^* + B_1^* B_1  &= \kbordermatrix{ 
& a & b & c & d & d & e & f\\
a & 2   & -1    & 0   & -1    & 0    & 0    & 0\\
b & -1    & 2   & -1    & 0   & -1    & 0   & -1\\
c & 0   & -1    & 2   & -1    & 1    & 0    & 1\\
d & -1    & 0   & -1    & 2    & 0    & 0    & 0\\
e & 0   & -1    & 1    & 0    & 3    & 0    & 0\\
f & 0    & 0    & 0    & 0    & 0    & 3    & 0\\
g & 0   & -1    & 1    & 0    & 0    & 0    & 3
}\\
H_2 = A_2 A_2^* + B_2^* B_2 &=
 \kbordermatrix{ 
& a & b & c & d & d & e & f\\
a & 2   & -1    & 0   & -1    & 0    & 0    & 1\\
b & -1    & 2   & -1    & 0   & -1    & 0   & -1\\
c & 0   & -1    & 2   & -1    & 1   & -1    & 0\\
d & -1    & 0   & -1    & 2    & 0    & 1    & 0\\
e & 0   & -1    & 1    & 0    & 2    & 0    & 0\\
f & 0    & 0   & -1    & 1    & 0    & 2   & -1\\
g & 1   & -1    & 0    & 0    & 0   & -1    & 2
}
\end{align*}
and verify that they have different spectra, as we had claimed in Example~\ref{eg:gip},
\[
\lambda(H_1) = (
    0, \;
    3-\sqrt{5}, \;
    2, \;
    3, \;
    3, \;
    3, \;
    3+\sqrt{5} ) \ne (
   0, \;
   0, \;
   3-\sqrt{5}, \;
    2, \;
    3, \;
    3, \;
    3+\sqrt{5} ) = \lambda(H_2).
\]

We now repeat the routine and convert the undirected graphs in Figure~\ref{fig:iso2} into labeled directed graphs $G_3$ (left) and $G_4$ (right) in Figure~\ref{fig:iso2a}.
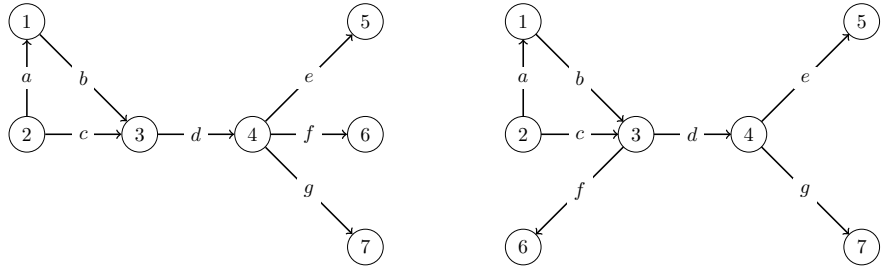
\begin{figure}[h]
\centering
\scalebox{0.75}{
\begin{tikzpicture}
\SetGraphUnit{2}
\GraphInit[vstyle=Dijkstra]
\Vertex{2}
\NO(2){1} \EA(2){3}
\EA(3){4}
\NOEA(4){5} \EA(4){6} \SOEA(4){7}
\tikzset{EdgeStyle/.style = {->}}
\Edge[label=$a$](2)(1)
\Edge[label=$b$](1)(3)
\Edge[label=$c$](2)(3)
\Edge[label=$d$](3)(4)
\Edge[label=$e$](4)(5)
\Edge[label=$f$](4)(6)
\Edge[label=$g$](4)(7)
\end{tikzpicture}}
\qquad\qquad
\scalebox{0.75}{
\begin{tikzpicture}
\SetGraphUnit{2}
\Vertex{2}
\NO(2){1} \EA(2){3}
\EA(3){4} \SOWE(3){6}
\NOEA(4){5} \SOEA(4){7}
\tikzset{EdgeStyle/.style = {->}}
\Edge[label=$a$](2)(1)
\Edge[label=$b$](1)(3)
\Edge[label=$c$](2)(3)
\Edge[label=$d$](3)(4)
\Edge[label=$e$](4)(5)
\Edge[label=$f$](3)(6)
\Edge[label=$g$](4)(7)
\end{tikzpicture}}
\caption{Labeled directed versions of the graphs in Figure~\ref{fig:iso2}.}
\label{fig:iso2a}
\end{figure}
We label both triangles in $G_3$ and $G_4$ as $T$ and orient it clockwise $2 \to 1 \to 3 \to 2$, giving us a matrix that represents both curl operators on $G_3$ and $G_4$,
\[
B_3 = B_4 =
\kbordermatrix{
& a & b & c & d & e & f & g\\
T & 1 & 1& -1& 0& 0& 0& 0
}.
\]

With these choices, we obtain the following matrix representations of the gradients, Laplacians, and Helmholtzians on $G_3$ and $G_4$, \enlargethispage{1.5\baselineskip}
\begin{align*}
A_3 &= \kbordermatrix{ 
& 1 & 2 &3 & 4 &5 &6 & 7\\
a & 1& -1& 0& 0& 0& 0& 0\\
b & -1& 0& 1& 0& 0& 0& 0\\
c & 0& -1& 1& 0& 0& 0& 0\\
d & 0& 0& -1& 1& 0& 0& 0\\
e & 0& 0& 0& -1& 1& 0& 0\\
f & 0& 0& 0& -1& 0& 1& 0\\
g & 0& 0& 0& -1& 0& 0& 1
},
\\
A_4 &=  \kbordermatrix{ 
& 1 & 2 &3 & 4 &5 &6 & 7\\
a & 1& -1& 0& 0& 0& 0& 0\\
b & -1& 0& 1& 0& 0& 0& 0\\
c & 0& -1& 1& 0& 0& 0& 0\\
d & 0& 0& -1& 1& 0& 0& 0\\
e & 0& 0& 0& -1& 1& 0& 0\\
f & 0& 0& -1& 0& 0& 1& 0\\
g & 0& 0& 0& -1& 0& 0& 1
},\\
L_3 = A_3^* A_3 &=  \kbordermatrix{ 
& 1 & 2 &3 & 4 &5 &6 & 7\\
1 & 2   & -1   & -1    & 0    & 0    & 0    & 0\\
2 & -1    & 2   & -1    & 0    & 0    & 0    & 0\\
3 & -1   & -1    & 3   & -1    & 0    & 0    & 0\\
4 & 0    & 0   & -1    & 4   & -1   & -1   & -1\\
5 & 0    & 0    & 0   & -1    & 1    & 0    & 0\\
6 & 0    & 0    & 0   & -1    & 0    & 1    & 0\\
7 & 0    & 0    & 0   & -1    & 0    & 0    & 1
},
\\
L_4 = A_4^* A_4 &=  \kbordermatrix{ 
& 1 & 2 &3 & 4 &5 &6 & 7\\
1 & 2   & -1   & -1    & 0    & 0    & 0    & 0\\
2 & -1    & 2   & -1    & 0    & 0    & 0    & 0\\
3 & -1   & -1    & 4   & -1    & 0   & -1    & 0\\
4 & 0    & 0   & -1    & 3   & -1    & 0   & -1\\
5 & 0    & 0    & 0   & -1    & 1    & 0    & 0\\
6 & 0    & 0   & -1    & 0    & 0    & 1    & 0\\
7 & 0    & 0    & 0   & -1    & 0    & 0    & 1
},\\
H_3 = A_3 A_3^* + B_3^* B_3 &= \kbordermatrix{
& a & b & c & d &  e & f & g\\
a & 3    & 0    & 0    & 0    & 0    & 0    & 0\\
b & 0    & 3    & 0   & -1    & 0    & 0    & 0\\
c & 0    & 0    & 3   & -1    & 0    & 0    & 0\\
d & 0   & -1   & -1    & 2   & -1   & -1   & -1\\
e & 0    & 0    & 0   & -1    & 2    & 1    & 1\\
f & 0    & 0    & 0   & -1    & 1    & 2    & 1\\
g & 0    & 0    & 0   & -1    & 1    & 1    & 2
},
\\
H_4 = A_4 A_4^* + B_4^* B_4 &=  \kbordermatrix{
& a & b & c & d & e & f & g\\
a & 3    & 0    & 0    & 0    & 0    & 0    & 0\\
b & 0    & 3    & 0   & -1    & 0   & -1    & 0\\
c & 0    & 0    & 3   & -1    & 0   & -1    & 0\\
d & 0   & -1   & -1    & 2   & -1    & 1   & -1\\
e & 0    & 0    & 0   & -1    & 2    & 0    & 1\\
f & 0   & -1   & -1    & 1    & 0    & 2    & 0\\
g & 0    & 0    & 0   & -1    & 1    & 0    & 2
}.
\end{align*}

As we intend to show that $G_3$ and $G_4$ have isospectral Hodge $k$-Laplacians for all $k$, we will also need to examine the Hodge $2$-Laplacian $\Delta_2$. Since $G_3$ and $G_4$ have no cliques of order higher than two, $\delta_k = 0$ for all $k > 2$ and in particular $\Delta_2 = \delta_1 \delta_1^*$. So the $1 \times 1$ matrix representing $\Delta_2$ is just
\[
P_3 \coloneqq  B_3 B_3^* = [3] = B_4 B_4^* =: P_4
\]
for both $G_3$ and $G_4$.

Finally, we verify that the spectra of the Hodge $k$-Laplacians of $G_3$ and $G_4$ are identical for $k = 0, 1, 2$, as we had claimed in Example~\ref{eg:gip}:
\begin{gather*}
\lambda(L_3) =(
   0, \;
    0.40, \;
    1, \;
    1, \;
    3, \;
    3.34, \;
    5.26)  = \lambda(L_4),\\
\lambda(H_3) = (
    0.40, \;
    1, \;
    1, \;
    3, \;
    3, \;
    3.34, \;
    5.26 ) = \lambda(H_4),\\
\lambda(P_3) = 3 = \lambda(P_4).
\end{gather*}
Observe that three eigenvalues of $L_3, L_4, H_3, H_4$ have been rounded to two decimal places --- these eigenvalues have closed form expressions (zeros of a cubic polynomial) but they are unilluminating and a hassle to typeset. So to verify that they are indeed isospectral, we check their characteristic polynomials instead, as these have integer coefficients and can be expressed exactly:
\begin{align*}
\det(L_3 -x I) &= -21 x + 112 x^2 - 209 x^3 + 178 x^4 - 73 x^5 + 14 x^6 - x^7\\
&=-x(x-3 ) (x-1)^2 (x^3 - 9x^2 + 21x -7) = \det(L_4 -x I),\\
\det(H_3 - xI) &= 63 - 357 x + 739 x^2 - 743 x^3 + 397 x^4 - 115 x^5 + 17 x^6 - x^7 \\
&= -(x-3 )^2 (x-1)^2 (x^3 - 9x^2 + 21x -7) = \det(H_4 - xI).
\end{align*}

\section{Topology, computations, and applications}

We conclude our article with this final section that (a) highlights certain deficiencies of our simplistic approach and provides pointers for further studies   (Section~\ref{sec:caveats}); (b) discusses how one may compute the quantities in this article using standard numerical linear algebra (Section~\ref{sec:compute}); and (c) proffers some high-level thoughts about applications to the information sciences (Section~\ref{sec:app}).

\subsection{Topological caveats}\label{sec:caveats}

The way we defined cohomology in Section~\ref{sec:coho} is more or less standard. The only simplification is that we had worked over a field. The notion of cohomology in topology works more generally over arbitrary rings where our simple linear algebraic approach falls short, but not by much --- all we need is to be willing to work with modules over rings instead of modules over fields, i.e., vector spaces. Unlike a vector space, a module may not have a basis and we may not necessarily be able to represent linear maps by matrices, a relatively small price to pay.

However the further simplifications in Sections~\ref{sec:harm} and \ref{sec:hodge} to avoid quotient spaces only hold when we have a field of characteristic zero (we chose $\mathbb{R}$). For example, if instead of $\mathbb{R}$, we had the field $\mathbb{F}_2$ of two elements with binary arithmetic (or indeed any field of positive characteristic), then we can no longer define inner products and statements like $\ker(B)^\perp = \im(B^*)$ make no sense. While the adjoint of a matrix may still be defined without reference to an inner product, statements like $\ker(A^* A) = \ker(A)$ are manifestly false in positive characteristic, as we saw at the end of Section~\ref{sec:NLA}.

We mentioned in Section~\ref{sec:terms} that in the way we presented things, there is no difference between cohomology and homology. This is an artifact of working over a field. In general cohomology and homology are different and are related via the universal coefficient theorem \cite{Hatcher}.

From the perspective of topology, the need to restrict to fields of zero  characteristic like $\mathbb{R}$ and $\mathbb{C}$ is a big shortcoming. For example, one would no longer be able to detect `torsion' and thereby perform basic topological tasks like distinguishing between a circle and a Klein bottle, which is a standard utility of cohomology groups over rings or fields of positive characteristics. We may elaborate on this point if the reader is willing to accept on faith that the cohomology group $H^k(G)$ in \eqref{eq:cohogp}  may still be defined even (i) when $G$ is a manifold, and (ii) when we replace our field of scalars $\mathbb{R}$ by a ring of scalars $\mathbb{Z}$. We will denote these cohomology groups over $\mathbb{R}$ and $\mathbb{Z}$ by $H^k(G; \mathbb{R})$ and $H^k(G; \mathbb{Z})$ respectively. For the circle $S^1$,  techniques standard in algebraic topology \cite{Hatcher} but beyond the scope of this article allow us to compute these:
\begin{equation}\label{eq:circle}
H^k(S^1; \mathbb{Z}) = 
\begin{cases}
\mathbb{Z} & k = 0,\\
\mathbb{Z} & k = 1,\\
0 & k \ge 2,
\end{cases}
\qquad
H^k(S^1; \mathbb{R}) = 
\begin{cases}
\mathbb{R} & k = 0,\\
\mathbb{R} & k = 1,\\
0 & k \ge 2.
\end{cases}
\end{equation}
Likewise, for the Klein bottle $K$, one gets
\begin{equation}\label{eq:klein}
H^k(K; \mathbb{Z}) = 
\begin{cases}
\mathbb{Z} & k = 0,\\
\mathbb{Z} & k = 1,\\
\mathbb{Z}_2 & k = 2,\\
0 & k \ge 3,
\end{cases}
\qquad
H^k(K; \mathbb{R}) = 
\begin{cases}
\mathbb{R} & k = 0,\\
\mathbb{R} & k = 1,\\
0 & k \ge 2.
\end{cases}
\end{equation}
Here $\mathbb{Z}_2 = \{0,1\}$ with addition performed modulo $2$; and whenever a cohomology group contains a nonzero element of finite order,\footnote{Recall that the order of an element is the number of times it must be added to itself to get $0$; but if this is never satisfied we say it has infinite order. In $\mathbb{Z}_2$, $1+ 1 =0$ so $1$ has order two.} we say that it has \emph{torsion}. There can never be torsion in $H^k(G; \mathbb{R})$ since every nonzero real number has infinite order. As we can see from \eqref{eq:circle} and \eqref{eq:klein}, $S^1$ and $K$ have identical cohomology groups over $\mathbb{R}$, i.e., we cannot tell apart  a circle from a Klein bottle with cohomology over $\mathbb{R}$. On the other hand, since $H^2(K; \mathbb{Z}) = \mathbb{Z}_2 \ne 0 = H^2(S^1; \mathbb{Z})$, cohomology over $\mathbb{Z}$ allows us to tell them apart.

Despite the aforementioned deficiencies, if one is primarily interested in engineering and scientific applications, then  we believe that our approach in Sections~\ref{sec:ped}, \ref{sec:graph}, and \ref{sec:HO} is by-and-large adequate. Furthermore, even though we have restricted our discussions in Sections~\ref{sec:graph} and \ref{sec:HO} to clique complexes of graphs, they apply verbatim to any simplicial complex.

We should add that although we did not discuss it, one classical use of cohomology and Hodge theory is to deduce topological information about an underlying topological space. Even over a field of characteristic zero, if we sample sufficiently many points $V$ from a sufficiently nice metric space $\Omega$, and set $G = (V,E)$ to be an appropriately chosen nearest neighbor graph, then
\begin{equation}\label{eq:betti}
\beta_k (G) = \dim H^k (G) = \dim \ker(\Delta_{k} )
\end{equation}
gives the number of `$k$-dimensional holes' in $\Omega$, called the \emph{Betti number}. While the kernel or $0$-eigenspace captures qualitative topological information, the nonzero eigenspaces often capture quantitative geometric information. In the context of graphs \cite{Chung, Spielman}, this is best seen in $\Delta_0$ --- its $0$-eigenpair tells us whether a graph is connected ($\beta_0(G)$ gives the number of connected components of $G$, as we saw in Section~\ref{sec:Helm}) while its smallest nonzero eigenpair tells us how connected the graph is (eigenvalue by way of the Cheeger inequality \cite[p.~26]{Chung} and eigenvector by way of the Fiedler vector \cite{Pothen}).

\subsection{Computations}\label{sec:compute}

A noteworthy point is that the quantities appearing in Sections~\ref{sec:graph} and \ref{sec:HO} are all computationally tractable\footnote{Although other related problems with additional conditions can be NP-hard \cite{DHK}.} and may in fact be computed using standard numerical linear algebra. In particular, the Hodge decomposition \eqref{eq:hohd} can be efficiently computed by solving least squares problems, which among other things gives us Betti numbers via \eqref{eq:betti} \cite{F}.  For  simplicity we will use the basic case in Section~\ref{sec:graph}  for illustration but the discussions below apply almost verbatim to the higher order cases in Section~\ref{sec:HO} as well.

Since $V$ is a finite set, $L^2(V)$, $L^2_\wedge(E)$, $L^2_\wedge(T)$ are finite-dimensional vector spaces. We may choose bases on these spaces 
and get
\[
L^2(V) \cong \mathbb{R}^p, \quad L^2_\wedge(E) \cong \mathbb{R}^n, \quad L^2_\wedge(T) \cong \mathbb{R}^m
\]
where $p,n,m$ are respectively the number of vertices, edges, and triangles in $G$. See also Section~\ref{sec:cal} for examples of how one may in practice write down matrices representing $k$-coboundary operators and Hodge $k$-Laplacians for $k = 0, 1, 2$.

Once we have represented cochains as vectors in Euclidean spaces, to compute the Helmholtz decomposition in \eqref{eq:helm1} for any given $X \in L^2_\wedge(E)$, we may solve the two least squares problems
\[
\min_{f \in L^2(V)} \| \grad f - X \| \qquad\text{and}\qquad \min_{\Phi \in L^2_\wedge(T)} \| \curl^* \Phi - X \|,
\]
to get $X_H$ as $X - \grad f - \curl^* \Phi$. Alternatively, we may solve
\[
 \min_{Y \in L_\wedge^2(E)} \| \Delta_1 Y - X \|
\]
for the minimizer $Y$ and get $X_H$ as the residual $X - \Delta_1 Y $ directly. Having obtained $\Delta_1 Y$, we may  use the decomposition \eqref{eq:sum},
\[
\setlength{\arraycolsep}{1.5pt}
\begin{matrix}
\im(\Delta_1) &= &  \im(\grad) &\oplus & \im(\curl^*),\\[0.5ex]
\Delta_1 Y & = & \grad f & + & \curl^* \Phi,
\end{matrix}
\]
and solve either
\[
\min_{f \in L^2(V)} \| \grad f - \Delta_1 Y \| \qquad\text{or}\qquad \min_{\Phi \in L^2_\wedge(T)} \| \curl^* \Phi - \Delta_1 Y \|
\]
to get the remaining two components.

We have the choice of practical, efficient, and stable methods like Krylov subspace methods for singular symmetric least squares problems \cite{CPS, CS} or specialized methods for the Hodge $1$-Laplacian with proven complexity bounds \cite{Miller}.

\subsection{Applications}\label{sec:app}

Traditional applied mathematics largely involves using partial differential equations to model physical phenomena and traditional computational mathematics largely revolves around numerical solutions of PDEs.

However, one usually needs substantial and rather precise knowledge about a phenomenon in order to write it down as PDEs. For example, one may need to know the dynamical laws (e.g., laws of motions, principle of least action, laws of thermodynamics, quantum mechanical postulates, etc) or conservation laws (e.g., of energy, momentum, mass, charge, etc) underlying the phenomenon before being able to `write down' the corresponding PDEs (as equations of motion, of continuity and transfer, constitutive equations, field equations, etc). In traditional applied mathematics, it is often taken for granted that there are known physical laws behind the phenomena being modeled.

In modern data applications, this is often a luxury. For example, if we want to build a spam filter, then it is conceivable that we would want to understand the `laws of emails.' But we would quickly come to the realization that these `laws of emails' would be too numerous to enumerate and too inexact to be formulated precisely, even if we restrict ourselves to those relevant for identifying spam. This is invariably the case for any human generated data: movie ratings, restaurant reviews, browsing behavior, clickthrough rates,  newsfeeds, tweets, blogs, instagrams, status updates on various social media, etc, but  it also applies to data from biology and medicine \cite{Lenoir}.

For such data sets, all one has is often a rough measure of how similar two data points are and how the data set is distributed. Topology can be a useful tool in such contexts \cite{C} since it requires very little --- essentially just a weak notion of separation, i.e., is there a non-trivial open set that contains those points?

If the data set is discrete and finite, which is almost always the case in applications, we can even limit ourselves to simplicial topology, where the topological spaces are simplicial complexes (see  Section~\ref{sec:cliques}). Without too much loss of generality, these may be regarded as clique complexes of graphs (see Section~\ref{sec:caveats}): data points are vertices in $V$ and proximity is characterized by cliques: a pair of data points are near each other  iff they form an edge in $E$, a triplet of data points are near one another iff they form a triangle in $T$, and so on.

In this article we have implicitly regarded a graph as a discrete analogue of a Riemannian manifold and cohomology as a discrete analogue of PDEs: standard partial differential operators on Riemannian manifolds --- gradient, divergence, curl, Jacobian, Hessian, scalar and vector Laplace operators, Hodge Laplacians ---  all have natural counterparts on graphs. An example of a line of work that carries this point of view to great fruition may be found in \cite{BenLov, BenSch, Lovasz}. Also, in this article we have only scratched the surface of cohomological and Hodge theoretic techniques in graph theory; see \cite{Hanlon} for results that go much further.

In traditional computational mathematics, discrete PDEs arise as discretizations of continuous PDEs, intermediate by-products of numerical schemes, and this accounts for the appearance of cohomology  in numerical analysis \cite{feec, AFW, DGLZ}. But in data analytic applications, discrete PDEs tend to play a more central and direct role. The discrete partial differential operators on graphs introduced in this article may perhaps serve as a bridge on which insights from traditional applied and computational mathematics could cross over and be brought to bear on modern data analytic applications. Indeed we have already begun to see some applications of the Hodge Laplacian and Hodge decomposition on graphs to:
\begin{enumerate}[\upshape (i)]
\item ranking \cite{Austin, HKW, JLYY, ODO, XHJYYL}, 
\item graphics and imaging \cite{DHLM, MMOC, TLHD},
\item games and traffic flows \cite{CMOP, COP},
\item brain networks \cite{brain},
\item data representations \cite{data},
\item deep learning \cite{deep},
\item denoising \cite{denoise},
\item dimension reduction \cite{dim},
\item link prediction \cite{link},
\item object synchronization \cite{synchron},
\item sensor network coverage \cite{coverage},
\item cryo-electron microscopy \cite{cryo},
\item generalizing effective resistance to simplicial complexes \cite{resist},
\item modeling biological interactions  between a set of molecules or communication systems with group messages \cite{group}.
\end{enumerate}

\section*{Acknowledgments}
The author owes special thanks to Gunnar Carlsson and Vin De Silva for introducing him to the subject many years ago. He also gratefully acknowledges exceptionally helpful discussions with Sayan Mukherjee and Yuan Yao, as well as the thorough reviews and many pertinent suggestions from the two anonymous referees that led to a much improved article.


\begin{thebibliography}{99}
\bibitem{feec} D.~Arnold, \emph{Finite element exterior calculus}, CBMS-NSF Regional Conference Series in Applied Mathematics, \textbf{93}, SIAM, Philadelphia, PA, 2018.

\bibitem {AFW} D.~Arnold, R.~S.~Falk, and R.~Winther, ``Finite element exterior calculus: from Hodge theory to numerical stability,'' \textit{Bull.\ Amer.\ Math.\ Soc.}, \textbf{47} (2010), no.~2, pp.~281--354.

\bibitem {Austin} D.~Austin, ``Who's number 1? Hodge theory will tell us,'' \textit{AMS Feature Column}, December 2012, \url{http://www.ams.org/samplings/feature-column/fc-2012-12}.

\bibitem{babai} L.~Babai, ``Automorphism groups, isomorphism, reconstruction,'' Chapter 27, pp.~1447--1540, in R.~L.~Graham, M.~Gr\"{o}tschel, and L.~Lov\'{a}sz (Eds), \textit{Handbook of Combinatorics}, \textbf{2}, Elsevier, Amsterdam, Netherlands, 1995.

\bibitem{babai2} L.~Babai, ``Graph isomorphism in quasipolynomial time,'' \textit{Proc.\ Annual ACM Symp.\ Theory Comput.} (STOC `16), \textbf{48}, pp.~684--697.

\bibitem {BSSS} L.~Bartholdi, T.~Schick, N.~Smale, and S.~Smale, ``Hodge theory on metric spaces,'' \textit{Found.\ Comput.\ Math.}, \textbf{12} (2012), no.~1, pp.~1--48.

\bibitem{BenLov} I.~Benjamini and L.~Lov\'{a}sz, ``Harmonic and analytic functions on graphs,'' \textit{J.\ Geom.}, \textbf{76} (2003), no.~1--2, pp.~3--15.

\bibitem{BenSch} I.~Benjamini and O.~Schramm, ``Harmonic functions on planar and almost planar graphs and manifolds, via circle packings,'' \textit{Invent.\ Math.}, \textbf{126} (1996), no.~3, pp.~565--587.

\bibitem{link} A.~R.~Benson, R.~Abebe, M.~T.~Schaub, A.~Jadbabaie, and J.~Kleinberg, ``Simplicial closure and higher-order link prediction,'' \textit{Proc.\ Natl.\ Acad.\ Sci.}, \textbf{115} (2018), no.~48, pp.~E11221--E11230.

\bibitem{bumper} J.~Bentley, ``Bumper-sticker computer science,'' \textit{Commun.\ ACM}, \textbf{28} (1985), no.~9, pp.~896--901.

\bibitem{deep} M.~M.~Bronstein, J.~Bruna, Y.~LeCun, A.~Szlam, and P.~Vandergheynst, ``Geometric deep learning: going beyond Euclidean data,'' \textit{IEEE Signal Proc.\ Mag.}, \textbf{34} (2017), no.~4, pp.~18--42.

\bibitem{butler} L.~M.~Butler, ``A beautiful mind,'' \textit{ Notices Amer.\ Math.\ Soc.}, \textbf{49} (2002), no.~4, pp.~455--457.

\bibitem{BG} S.~Butler and J.~Grout, ``A construction of cospectral graphs for the normalized Laplacian,'' \textit{Electron.\ J.\ Combin.}, \textbf{18} (2011), no.~1, Paper 231, 20 pp. 

\bibitem {CMOP} O.~Candogan, I.~Menache, A.~Ozdaglar, and P.~Parrilo, ``Flows and decompositions of games: harmonic and potential games,'' \textit{Math.\ Oper.\ Res.}, \textbf{36} (2011), no.~3, pp.~474--503.

\bibitem {COP} O.~Candogan, A.~Ozdaglar, and P.~Parrilo, ``Dynamics in near-potential games,'' \textit{Games Econom.\ Behav.}, \textbf{82} (2013), pp.~66--90. 

\bibitem{C} G.~Carlsson, ``Topological pattern recognition for point cloud data,'' \textit{Acta Numer.}, \textbf{23} (2014), pp.~289--368. 

\bibitem{data} C.~K.~Chui, H.~N.~Mhaskar, and X.~Zhuang, ``Representation of functions on big data associated with directed graphs,'' \emph{Appl.\ Comput.\ Harmon.\ Anal.}, \textbf{44} (2018), no.~1, pp.~165--188.

\bibitem{CPS}  S.-C.~T.~Choi, C.~C.~Paige, and M.~A.~Saunders, ``MINRES-QLP: a Krylov subspace method for indefinite or singular symmetric systems,'' \textit{SIAM J.\ Sci.\ Comput.}, \textbf{33} (2011), no.~4, pp.~1810--1836.

\bibitem{CS}  S.-C.~T.~Choi and M.~A.~Saunders, ``Algorithm 937: MINRES-QLP for symmetric and Hermitian linear equations and least-squares problems,'' \textit{ACM Trans.\ Math.\ Software},  \textbf{40} (2014), no.~2, Art.~16, 12 pp.

\bibitem{Chung} F.~R.~K.~Chung, \textit{Spectral Graph Theory}, CBMS Regional Conference Series in Mathematics, \textbf{92}, AMS, Providence, RI, 1997.

\bibitem{Miller} M.~B.~Cohen, B.~T.~Fasy, G.~L.~Miller, A.~Nayyeri, R.~Peng, and N.~Walkington, ``Solving $1$-Laplacians of convex simplicial complexes in nearly linear time,'' \textit{Proc.\ Annual ACM--SIAM Symp.\ Discrete Algorithms} (SODA '14), \textbf{25} (2014), pp.~204--216.

\bibitem {DMV} V.~De~Silva, D.~Morozov, and M.~Vejdemo-Johansson, ``Persistent cohomology and circular coordinates,'' \textit{Discrete Comput.\ Geom.}, \textbf{45} (2011), no.~4, pp.~737--759.

\bibitem{DHLM} M.~Desbrun, A.~Hirani, M.~Leok, and J.~E.~Marsden, ``Discrete exterior calculus,'' \textit{preprint}, (2005), \url{http://arxiv.org/abs/math/0508341}.

\bibitem{DLM} M.~Desbrun, M.~Leok, J.~E.~Marsden, ``Discrete Poincar\'e lemma,''
\textit{Appl.\ Numer.\ Math.}, \textbf{53} (2005), no.~2--4, pp.~231--248. 

\bibitem {DHK} T.~Dey, A.~Hirani, and B.~Krishnamoorthy, ``Optimal homologous cycles, total unimodularity, and linear programming,'' \textit{SIAM J.\ Comput.}, \textbf{40} (2011), no.~4, pp.~1026--1044.

\bibitem {D1} J.~Dodziuk, ``Finite-difference approach to the Hodge theory of harmonic forms,'' \textit{Amer.\ J.\ Math.}, \textbf{98} (1976), no.~1, pp.~79--104.

\bibitem {DGLZ} Q.~Du, M.~Gunzburger, R.~Lehoucq, and K.~Zhou, ``A nonlocal vector calculus, nonlocal volume-constrained problems, and nonlocal balance laws,'' \textit{Math.\ Models Methods Appl.\ Sci.}, \textbf{23} (2013), no.~3, pp.~493--540. 

\bibitem{Eck} B.~Eckmann, ``Harmonische funktionen und randwertaufgaben in einem komplex,'' \textit{Comment.\ Math.\ Helv.}, \textbf{17}, (1945), pp.~240--255.

\bibitem {F} J.~Friedman, ``Computing Betti numbers via combinatorial Laplacians,'' \textit{Algorithmica}, \textbf{21} (1998), no.~4, pp.~331--346.

\bibitem{synchron} T.~Gao, J.~Brodzki, and S.~Mukherjee, ``The geometry of synchronization problems and learning group actions,'' \textit{Discrete Comput.\ Geom.}, (2019), to appear.

\bibitem{GMW} P.~E.~Gill, W.~Murray, and M.~H.~Wright, \textit{Numerical Linear Algebra and Optimization}, Addison--Wesley, Redwood City, CA, 1991.

\bibitem {GLS} J.-M.~Guo, J.~Li, and W.~C.~Shiu, ``On the Laplacian, signless Laplacian and normalized Laplacian characteristic polynomials of a graph,'' \textit{Czechoslovak Math.\ J.}, \textbf{63} (2013), no.~3, pp.~701--720. 

\bibitem{Hanlon} P.~Hanlon, ``The Laplacian method,'' pp.~65--91, in S.~Fomin (Ed), \textit{Symmetric Functions 2001: Surveys of Developments and Perspectives},
NATO Science Series II: Mathematics, Physics and Chemistry, \textbf{74}, Kluwer, Dordrecht, Netherlands, 2002.

\bibitem{Hatcher} A.~Hatcher, \textit{Algebraic Topology}, Cambridge University Press, Cambridge, UK, 2002.

\bibitem{Hira} A.~Hirani, \textit{Discrete Exterior Calculus}, Ph.D.\ thesis, California Institute of Technology, Pasadena, CA, 2003.

\bibitem {HKW} A.~Hirani, K.~Kalyanaraman, and S.~Watts, ``Least squares ranking on graphs,'' \textit{preprint}, (2011), \url{http://arxiv.org/abs/1011.1716}.

\bibitem{H} W.~V.~D.~Hodge, \textit{The Theory and Applications of Harmonic Integrals}, 2d Ed., Cambridge University Press, Cambridge, 1952.

\bibitem {JLYY} X.~Jiang, L.-H.~Lim, Y.~Yao, and Y.~Ye, "Statistical ranking and combinatorial Hodge theory," \textit{Math.\ Program.}, \textbf{127} (2011), no.~1, pp.~203--244.

\bibitem{Kac} M.~Kac, ``Can one hear the shape of a drum?'' \textit{Amer.\ Math.\ Monthly}, \textbf{73} (1966), no.~4, pp.~1--23. 

\bibitem {KE} P.~Kingston and M.~Egerstedt, ``Distributed-infrastructure multi-robot routing using a Helmholtz--Hodge decomposition,'' \textit{Proc.\ IEEE Conf.\ Decis.\ Control} (CDC),
\textbf{50} (2011), pp.~5281--5286.

\bibitem{resist} W.~Kook and K.-J.~Lee, ``Simplicial networks and effective resistance,'' \textit{Adv.\ in Appl.\ Math.}, \textbf{100} (2018), pp.~71--86.

\bibitem{brain} H.~Lee et al., ``Harmonic holes as the submodules of brain network and network dissimilarity,'' pp. 110--122 in: R.~Marfil et al.\ (Eds), \emph{Computational Topology in Image Context}, Lecture Notes in Computer Science, \textbf{11382}, Springer, Cham, 2019.

\bibitem{Lenoir} T.~Lenoir, ``Shaping biomedicine as an information science,''
pp.~27--45, in M.~E.~Bowden, T.~B.~Hahn, and R.~V.~Williams (Eds), \emph{Proc.\ Conf.\ Hist.\ Heritage Sci.\ Inform.\ Syst.}, ASIS Monograph Series, Medford, NJ, 1999.


\bibitem{ped1} H.~J.~Lipkin, \textit{Beta Decay for Pedestrians}, Dover, Mineola, NY, 2004.

\bibitem{ped2} H.~J.~Lipkin, \textit{Lie Groups for Pedestrians}, 2nd Ed., Dover, Mineola, NY, 2002.

\bibitem{Lovasz} L.~Lov\'{a}sz, ``Discrete analytic functions: an exposition,'' pp.~241--273, in A.~Grigor’yan and S.~T.~Yau (Eds), \textit{Surveys in Differential Geometry: Eigenvalues of Laplacians and Other Geometric Operators}, \textbf{IX}, International Press, Somerville, MA, 2004.

\bibitem{MMOC} W.~Ma, J.-M.~Morel, S.~Osher, and A.~Chien, ``An $L_1$-based variational model for retinex theory and its applications to medical images,''  \textit{Proc.\ IEEE Conf.\ Comput.\ Vis.\ Pattern Recognit.} (CVPR), \textbf{29} (2011), pp.~153--160.

\bibitem {ODO} B.~Osting, J.~Darbon, and S.~Osher, ``Statistical ranking using the $\ell^{1}$-norm on graphs,'' \textit{Inverse Probl.\ Imaging}, \textbf{7} (2013), no.~3, pp.~907--926.

\bibitem{dim} J.~A.~Perea,
``Multiscale projective coordinates via persistent cohomology of sparse filtrations,''
\textit{Discrete Comput.\ Geom.}, \textbf{59} (2018), no.~1, pp.~175--225.


\bibitem{Pothen} A.~Pothen, H.~D.~Simon, and K.-P.~Liou, ``Partitioning sparse matrices with eigenvectors of graphs,'' \textit{SIAM J.\ Matrix Anal.\ Appl.}, \textbf{11} (1990), no.~3, pp.~430--452.

\bibitem{group} M.~T.~Schaub, A.~R.~Benson, P.~Horn, G.~Lippner, and A.~Jadbabaie, ``Random walks on simplicial complexes and the normalized Hodge Laplacian,'' \textit{preprint}, (2018), \url{https://arxiv.org/abs/1807.05044}.

\bibitem{denoise} M.~T.~Schaub and S.~Segarra, ``Flow smoothing and denoising: graph signal processing in the edge-space,'' \textit{Proc.\ IEEE Global Conf.\ Signal Inform.\ Process.} (GlobalSIP), (2018), pp.~735--739.

\bibitem{DGC} H.~M.~Schey, \textit{Div, Grad, Curl, and All That: An informal text on vector calculus}, 4th Ed., Norton, New York, NY, 2005.

\bibitem{onefoot} B.~Simon, ``OPUC on one foot,'' \textit{Bull.\ Amer.\ Math.\ Soc.}, \textbf{42} (2005), no.~4, pp.~431--460.

\bibitem{SS} N.~Smale and S.~Smale, ``Abstract and classical Hodge–-de Rham theory,'' \textit{Anal.\ Appl.}, \textbf{10} (2012), no.~1, pp.~91--111.

\bibitem{Spielman} D.~Spielman, ``Spectral graph theory,'' Chapter 18, pp.~495--524, in U.~Naumann and O.~Schenk (Eds), \textit{Combinatorial Scientific Computing}, CRC Press, Boca Raton, FL, 2012.

\bibitem{Strang} G.~Strang, ``The fundamental theorem of linear algebra,'' \textit{Amer.\ Math.\ Monthly}, \textbf{100} (1993), no.~9, pp.~848--855.

\bibitem {TJ} A.~Tahbaz-Salehi and A.~Jadbabaie, ``Distributed coverage verification in sensor networks without location information,'' \textit{IEEE Trans.\ Autom.\ Control}, \textbf{55} (2010), no.~8, pp.~1837--1849.

\bibitem {TLHD} Y.~Tong, S.~Lombeyda, A.~Hirani, and M.~Desbrun, ``Discrete multiscale vector field decomposition,'' \textit{ACM Trans.\ Graph.},
\textbf{22} (2003), no.~3, pp.~445--452.

\bibitem{W} F.~W.~Warner, \textit{Foundations of Differentiable Manifolds and Lie Groups}, Graduate Texts in Mathematics, \textbf{94}, Springer-Verlag, New York, NY, 1983.

\bibitem {XHJYYL} Q.~Xu, Q.~Huang, T.~Jiang, B.~Yan, Y.~Yao, and W.~Lin, ``HodgeRank on random graphs for subjective video quality assessment,'' \textit{IEEE Trans.\ Multimedia}, \textbf{14} (2012), no.~3, pp.~844--857.

\bibitem{cryo} K.~Ye and L.-H.~Lim, ``Cohomology of cyro-electron microscopy,'' \textit{SIAM J.\ Appl.\ Algebra Geometry}, \textbf{1} (2017), no.~1, pp.~507--535.

\bibitem{coverage} M.~Zhang, A.~Goupil, A.~Hanaf, and T.~Wang, ``Distributed harmonic form computation,'' \textit{IEEE Signal Proc.\ Let.}, \textbf{25} (2018), no.~8, pp.~1241--1245.
\end{thebibliography}
\end{document}